\documentclass[11pt,onecolumn]{IEEEtran}

\usepackage[utf8]{inputenc}
\usepackage[T1]{fontenc}
\usepackage[pdftex,left=1in,top=1in,bottom=1in,right=1in]{geometry}
\usepackage{amsmath, amssymb, dsfont, mathrsfs, amsthm}
\usepackage{amsfonts}
\usepackage{enumitem}
\usepackage{mathtools}
\usepackage{bbm}
\usepackage{algorithm}
\usepackage{algpseudocode}
\usepackage{graphicx}
\usepackage{comment}
\usepackage{bm}
\usepackage{xcolor}
\usepackage[caption=false,font=footnotesize]{subfig}
\usepackage[colorlinks=true,citecolor=blue,linkcolor=blue]{hyperref}
\hypersetup{
    colorlinks,
    linkcolor={red!50!black},
    citecolor={blue!50!black},
    urlcolor={blue!80!black}
}
\usepackage[nameinlink, noabbrev, capitalize]{cleveref}
\crefname{claim}{claim}{claims}
\Crefname{claim}{Claim}{Claims}
\usepackage{tikz}
\usetikzlibrary{arrows.meta, positioning}
\usepackage{pgfplots}
\pgfplotsset{compat=1.18}
\definecolor{dred}{RGB}{170,0,0}

\definecolor{dblue}{RGB}{0,120,215}

\newcommand{\E}{\mathds{E}}

\newcommand{\cC}{\mathcal{C}}

\newcommand{\cF}{\mathcal{F}}
\newcommand{\cG}{\mathcal{G}}

\newcommand{\cJ}{\mathcal{J}}
\newcommand{\cK}{\mathcal{K}}

\newcommand{\cP}{\mathcal{P}}
\newcommand{\cS}{\mathcal{S}}
\newcommand{\cT}{\mathcal{T}}
\newcommand{\cU}{\mathcal{U}}
\newcommand{\cV}{\mathcal{V}}

\newcommand{\cY}{\mathcal{Y}}

\newcommand{\supp}{\mathsf{supp}}

\newcommand{\eps}{\varepsilon}
\newtheorem{theorem}{Theorem}

\newtheorem{claim}{Claim}
\newtheorem{lemma}{Lemma}

\newtheorem{definition}{Definition}
\newtheorem*{definition*}{Definition}
\newtheorem{remark}{Remark}

\newtheoremstyle{notationstyle}
  {3pt}
  {3pt}
  {\normalfont}
  {}
  {\itshape}
  {.}
  { }
  {}
\theoremstyle{notationstyle}
\newtheorem*{notation}{Notation}

\newcommand{\bits}{\{0,1\}}

\newcommand{\Enc}{\operatorname{Enc}}
\newcommand{\Dec}{\operatorname{Dec}}

\newcommand{\peff}{p_{\mathrm{eff}}}
\newcommand{\peffmax}{p_{\mathrm{eff}_{\text{max}}}}

\newcommand{\qerase}{q_{\mathrm{erase}}}

\newcommand{\snote}[1]{{\color{cyan} \footnotesize(Samuel: #1)}}
\newcommand{\knote}[1]{{\color{blue} \footnotesize(Keshav: #1)}}
\newcommand{\sknote}[1]{{\color{purple} \footnotesize(Serge: #1)}}

\newcommand{\Bin}{\mathsf{Bin}}

\newcommand{\Cbsc}{C_{\mathrm{2BSC}}}
\newcommand{\Ebsc}{E_{\mathrm{2BSC}}}
\newcommand{\twobsc}{\operatorname{2BSC}}
\newcommand{\Css}{C_{\textnormal{syn--seq}}}
\newcommand{\vecqp}{\vec{q}'}
\newcommand{\vecqpi}{\vec{q}_i'}
\newcommand{\vecep}{\vec{e}'}
\newcommand{\indep}{\perp\!\!\!\perp}

\newcommand{\vhz}{\hat{\vec{Z}}}
\newcommand{\hq}{\hat{Q}}
\renewcommand{\hm}{\hat{M}}
\renewcommand{\SS}{\mathsf{SS}}
\newcommand{\hw}{\widehat{W}}
\newcommand{\OSS}{\mathsf{OSS}}
\newcommand{\Err}{\mathsf{Err}}
\let\originalleft\left
\let\originalright\right
\renewcommand{\left}{\mathopen{}\mathclose\bgroup\originalleft}
\renewcommand{\right}{\aftergroup\egroup\originalright}
\newcommand{\psyn}{p_{\mathrm{syn}}}
\newcommand{\Qsyn}{Q_{\mathrm{syn}}}
\newcommand{\pseq}{p_{\mathrm{seq}}}
\newcommand{\Qseq}{Q_{\mathrm{seq}}}
\newcommand{\iid}{\overset{\mathrm{i.i.d.}}{\sim}}
\newcommand{\lambdasyn}{\lambda_{\mathrm{syn}}}
\newcommand{\lambdaseq}{\lambda_{\mathrm{seq}}}

\newcommand{\Chsyn}{\mathrm{Ch}_{\mathrm{syn}}}
\newcommand{\Chseq}{\mathrm{Ch}_{\mathrm{seq}}}

\renewcommand\vec{\mathbf}

\newcommand{\tF}{\tilde{F}}
\newcommand{\sA}{\mathsf{A}}
\newcommand{\sB}{\mathsf{B}}

\begin{document}

\title{The Synthesis--Sequencing Channel for DNA-based Data Storage}

\author{%
  Keshav Goyal,
  Samuel Pearson,
  Jo\~{a}o Ribeiro,
  and~Serge Kas Hanna%
  \thanks{The first two authors contributed equally to this work.}%
  \thanks{K.~Goyal and S.~Kas Hanna are with the C\^{o}te d'Azur University, CNRS, I3S Laboratory, 06900 Sophia Antipolis, France. E-mails: \{keshav.goyal, serge.kas-hanna\}@cnrs.fr.}%
  \thanks{S.~Pearson and J.~Ribeiro are with Instituto de Telecomunica\c{c}\~{o}es, 1049-001 Lisboa, Portugal, and with the Departamento de Matem\'{a}tica, Instituto Superior T\'{e}cnico, Universidade de Lisboa, 1049-001 Lisboa, Portugal. E-mails: \{samuelpearson, jribeiro\}@tecnico.ulisboa.pt.}%
}

\maketitle

\begin{abstract}
   We introduce and study the synthesis--sequencing channel, a two-stage model for DNA-based data storage that jointly captures synthesis and sequencing effects. The synthesis--sequencing channel provides a more nuanced and realistic model of the DNA storage process compared to prior work, as it distinguishes between physical coverage after synthesis and sequencing coverage after readout, relaxes the assumption of independent errors across reads, and naturally induces coverage bias through the composition of synthesis and sequencing stages. We establish the information-theoretic capacity of this channel by deriving matching converse and achievability bounds for the case where synthesis and sequencing errors are modeled by binary symmetric channels with possibly different error probabilities, under mild assumptions on the channel parameters. Our results reveal multiple trade-offs between physical coverage, synthesis errors, sequencing coverage, and sequencing errors that influence the maximum achievable rate for reliable data storage.
\end{abstract}

\section{Introduction}

%\jnote{throughout the paper we need to be consistent on whether we use \Cref{thm:cap-main} or Theorem~\ref{thm:cap-main}}

The exponential growth of digital data has challenged the limits of conventional storage technologies and motivated the exploration of alternative media to meet the growing demand. DNA has emerged as a particularly promising molecular medium for data storage due to its exceptional density (exabytes per gram of DNA) and long-term durability (hundreds of years)~\cite{church2012,grass2015robust}. The DNA storage pipeline consists of a writing process, where digital information is encoded into a collection of short quaternary sequences over the four DNA nucleotides ($\mathsf{A}$, $\mathsf{G}$, $\mathsf{C}$, $\mathsf{T}$), which are then synthesized as DNA strands. The synthetic DNA is subsequently stored in molecular form under suitable conditions for long-term preservation. During reading, a sample of the stored DNA is amplified, sequenced, and the resulting reads are processed to retrieve the digital information. Several experimental works have demonstrated the feasibility of this storage pipeline over the past decade~\cite{grass2015robust, goldman2013towards, yazdi2017portable, chandak2019improved, press2020hedges, welzel2023dna, bar2025scalable, DNA-MGCP}.

The DNA channel is characterized by multiple sources of randomness arising at different stages of the storage pipeline~\cite{heckel2019characterization,shomorony2022information}. These include: (1) variability in physical coverage during storage, stemming from the randomness of the synthesis process where unequal numbers of molecular copies are produced for each DNA sequence; (2) variability in sequencing coverage, that is, different numbers of reads per encoded sequence resulting from the combined effects of unequal physical coverage, PCR amplification bias, and the stochastic nature of the sequencing process; (3) the unordered nature of sequencing reads; and (4) random nucleotide-level errors introduced during synthesis, amplification, and sequencing due to biochemical noise. These unique characteristics distinguish DNA storage from conventional communication systems and have motivated the study of information-theoretic limits for channel models that capture, in part, these sources of uncertainty.

Existing work has largely focused on variants of the so-called noisy drawing channel, which models the DNA storage process through three main mechanisms: random sampling, a noisy communication channel, and random shuffling. More specifically, the channel input consists of multiple DNA sequences, and the output consists of noisy, unordered reads obtained by drawing sequences at random according to a sampling distribution \(Q\), transmitting each drawn sequence independently through a noisy memoryless channel with per-symbol error probability \(p\), and observing the resulting reads in shuffled order without knowledge of their origin.\footnote{Prior work has considered two sampling models: iterative sampling with replacement, and a model where the $i$-th sequence is sampled $Q_i$ times with the $Q_i$'s drawn i.i.d. from a common distribution $Q$. The latter model simplifies the analysis and approximates the former in the limit.}
Early works considered the noiseless version of this model under uniform independent sampling and established its capacity~\cite{heckel2017fundamental}. This was followed by extensions incorporating noise, including the case where each sequence is drawn exactly once and transmitted over a binary symmetric channel, for which capacity and efficient schemes were derived~\cite{shomorony2019capacity}. Subsequent work generalized these capacity results to probabilistic sampling models, such as Bernoulli drawing distributions~\cite{shomorony2021dna}. Additional contributions further extended the model to general memoryless channels and broader classes of sampling distributions, establishing capacity results under increasingly general assumptions~\cite{weinberger2022dna,lenz2022noisy}. More recently, achievability bounds have been derived in regimes where DNA sequences are shorter than those considered in earlier works~\cite{tamir2025dnastorageshortmolecule,TWG26}, and aspects beyond the capacity have also been studied, such as the achievable error exponent~\cite{Wei22,LS25,LWS25}.

In this work, we introduce and study a more general model of the DNA storage channel, referred to as the \emph{synthesis--sequencing} channel. This channel, illustrated in Fig.~\ref{fig:channel_model}, models the DNA storage pipeline as a two-stage process. The first stage corresponds to synthesis and includes both random sampling and noise, capturing variability in physical coverage as well as errors introduced during synthesis.\footnote{In practice, multiple DNA strands (molecules) are synthesized for each input sequence, and synthesis errors are becoming increasingly relevant with scalable biochemical technologies such as photolithographic synthesis~\cite{antkowiak2020low}.} This stage is characterized by a sampling distribution $\Qsyn$ and by symbol-level errors modeled through a discrete memoryless channel with per-symbol error probability $\psyn$, and its output corresponds to the molecular population present during the storage phase. The second stage jointly models amplification and sequencing during readout. It is characterized by a sampling distribution $\Qseq$, a discrete memoryless channel with error probability $\pseq$, and random shuffling of the resulting reads. We note that it is also of interest to study the synthesis--sequencing channel under more general noise models, and extending the framework to such settings is natural.

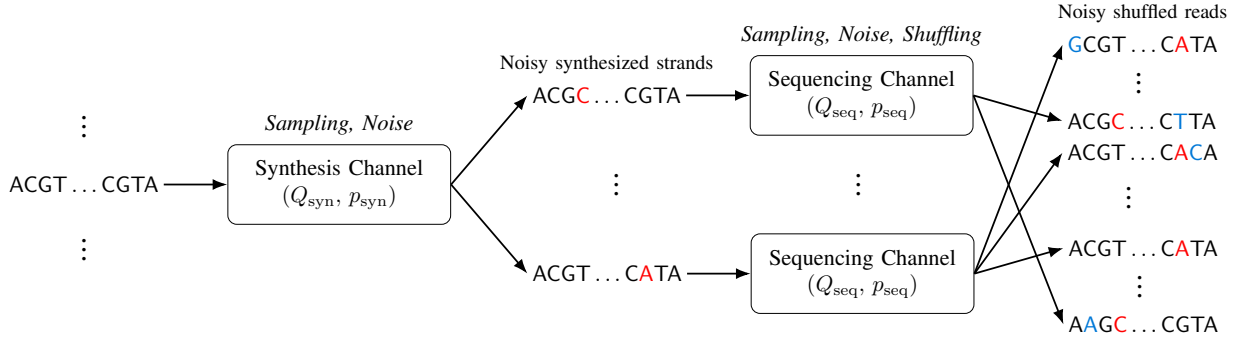
\begin{figure*}[h!]
\centering
\resizebox{\textwidth}{!}{
\begin{tikzpicture}[
    font=\small,
    >=Latex,
    node distance=1.2cm and 1.2cm,
    dna/.style={
        inner sep=2pt,
        align=left,
        anchor=west
    },
    block/.style={
        draw,
        rounded corners,
        minimum width=3.5cm,
        minimum height=1.25cm,
        align=center
    },
    lbl/.style={
        align=center
    }
]

% Input
\node[dna] (in) {$\mathsf{ACGT\ldots CGTA}$};
\node[above= 0.01 cm of in, xshift=-1.3cm, yshift=0.4cm, xshift=1.3cm] (dots1) {\Large$\vdots$};
\node[below= 0.01 cm of in, xshift=-1.3cm, yshift=-0.3cm, xshift=1.3cm] (dots1) {\Large$\vdots$};

% Synthesis block
\node[block, right=1.0cm of in] (syn) {Synthesis Channel\\$(\Qsyn,\, \psyn)$};
\node[lbl, above=0.01cm of syn] {{\em Sampling, Noise}};

% Output of synthesis
\node[dna, right=1.2cm of syn, yshift=1.4cm] (mol1) {$\mathsf{ACG{\color{red}C}\ldots CGTA}$};
\node[dna, right=1.2cm of syn, yshift=-1.4cm] (mol2) {$\mathsf{ACGT\ldots C{\color{red}A}TA}$};
\node[right=2.4cm of syn, yshift=0.1cm] (dots1) {\Large$\vdots$};
\node[lbl, above=0.01cm of mol1] {\footnotesize Noisy synthesized strands};

% Sequencing blocks
\node[block, right=1.0cm of mol1] (seq1) {Sequencing Channel\\$(\Qseq,\, \pseq)$};
\node[block, right=1.0cm of mol2] (seq2) {Sequencing Channel\\$(\Qseq,\, \pseq)$};
\node[lbl, above=0.01cm of seq1] {{\em Sampling, Noise, Shuffling}};
\node[right=6.2cm of syn, yshift=0.1cm] (dots1) {\Large$\vdots$};
\node[right=10.45cm of syn, yshift=-0.1cm] (dots1) {\Large$\vdots$};

% Reads
\node[dna, right=1.4cm of seq1, yshift=0.8cm] (r11) {$\mathsf{{\color{dblue}G}CGT\ldots C{\color{red}A}TA}$};
\node[dna, right=1.4cm of seq1, yshift=-0.4cm] (r12) {$\mathsf{ACG{\color{red}C}\ldots C{\color{dblue}T}TA}$};
\node[dna, right=1.4cm of seq1, yshift=-0.9cm] (r13) {$\mathsf{ACGT\ldots C{\color{red}A}{\color{dblue}C}A}$};
\node[right=2.4cm of seq1, yshift=0.3cm] (dots2) {\Large$\vdots$};
\node[lbl, above=0.01cm of r11] {\footnotesize Noisy shuffled reads};

\node[dna, right=1.4cm of seq2, yshift=0.4cm] (r21) {$\mathsf{ACGT\ldots C{\color{red}A}TA}$};
\node[dna, right=1.4cm of seq2, yshift=-0.8cm] (r22) {$\mathsf{A{\color{dblue}A}G{\color{red}C}\ldots CGTA}$};
\node[right=2.4cm of seq2,  yshift=-0.1cm] (dots3) {\Large$\vdots$};

% Arrows (FIXED)
\draw[->, thick] (in.east) -- (syn.west);

\draw[->, thick] (syn.east) -- (mol1.west);
\draw[->, thick] (syn.east) -- (mol2.west);

\draw[->, thick] (mol1.east) -- (seq1.west);
\draw[->, thick] (mol2.east) -- (seq2.west);

\draw[->, thick] (seq1.east) -- (r22.west);
\draw[->, thick] (seq1.east) -- (r12.west);
\draw[->, thick] (seq2.east) -- (r13.west);

\draw[->, thick] (seq2.east) -- (r21.west);
\draw[->, thick] (seq2.east) -- (r11.west);

\end{tikzpicture}
}
\caption{Schematic illustration of the DNA synthesis--sequencing channel. Red symbols indicate synthesis errors that appear systematically across corresponding reads, and blue symbols indicate sequencing errors. For the sake of exposition this figure focuses only on one input DNA sequence. In general, the channel takes multiple DNA sequences as input and outputs a collection of noisy reads that are shuffled across all sequences.}
\label{fig:channel_model}
\end{figure*}

Compared to the one-stage noisy drawing channel considered in prior work, the synthesis--sequencing channel provides a more accurate representation of the DNA storage process by capturing additional structural properties. In particular, it enables a more nuanced treatment of coverage by distinguishing between \emph{physical coverage} resulting from synthesis and \emph{sequencing} coverage after readout. This distinction is important since physical redundancy is inversely related to the achievable storage density~\cite{gimpel2026comparison}, and thus the model allows evaluating capacity trade-offs between high- and low-density operating regimes. Moreover, the model relaxes the common assumption of independent errors across reads, as synthesis errors propagate \emph{systematically} across all reads originating from the same strand (molecule). These systematic errors have a pronounced impact on capacity in the low physical redundancy (high storage density) regime, as quantified in our analysis. Finally, the composition of the two sampling stages naturally induces a positively skewed coverage distribution at the output, which is consistent with empirical observations~\cite{gimpel2023digital}.

We initiate the study of the capacity of synthesis--sequencing channels.
Our main technical contribution is establishing the capacity of the synthesis--sequencing channel in the case where both the synthesis and sequencing noisy channels are binary symmetric channels with possibly different error probabilities, under mild assumptions on the system's parameters.
% by deriving matching converse and achievability results under mild assumptions on the channel parameters. 
The resulting capacity expression reveals several non-trivial trade-offs that characterize the behavior of the DNA storage channel across different operating regimes. In particular, it highlights how the channel conditions related to synthesis and sequencing jointly influence the maximum achievable information rate for reliable data storage.
A formal statement of our result is provided in \Cref{sec:main-result} along with a numerical example that illustrates these trade-offs.

%To illustrate these trade-offs, we present a representative numerical example showing the capacity of the synthesis--sequencing channel under various conditions alongside that of the one-stage channel model considered in prior work, which we refer to as the \emph{sequencing-only} channel.

% Our main technical contribution is establishing the capacity of the synthesis--sequencing channel \jnote{we should state here that we do this for the BSC-BSC case} by deriving matching converse and achievability results under mild assumptions on the channel parameters. The resulting capacity expression reveals several non-trivial trade-offs that characterize the behavior of the DNA storage channel across different operating regimes. In particular, it highlights how the channel conditions related to synthesis and sequencing jointly influence the maximum achievable information rate for reliable data storage. To illustrate these trade-offs, we present a representative numerical example showing the capacity of the synthesis--sequencing channel under various conditions alongside that of the one-stage channel model considered in prior work, which we refer to as the \emph{sequencing-only} channel.

\section{Channel Model}\label{sec:model}

%\jnote{made significant edits, please review}

We now discuss our channel model more formally.
As mentioned above, a synthesis--sequencing channel is characterized by synthesis- and sequencing-level sampling distributions $\Qsyn$ and $\Qseq$, and by synthesis and sequencing discrete memoryless channels $\Chsyn$ and $\Chseq$.
The input to the synthesis--sequencing channel consists of $M$ sequences $X_1,\dots,X_M$ over a finite alphabet $\Sigma$, with each $X_i$ having length~$L$. 
At the synthesis stage, each $X_i$ is drawn $Q_{1,i}$ times, with $Q_{1,1},\dots,Q_{1,M} \iid \Qsyn$, and each of the $Q_{1,i}$ drawn sequences is independently transmitted through $\Chsyn$, yielding the noisy synthesized strands $Z_{i,1},\dots,Z_{i,Q_{1,i}}$. At the sequencing stage, each $Z_{i,j}$ is independently drawn $Q_{2,i,j} \sim \Qseq$ times, and each of the drawn sequences is independently transmitted through $\Chseq$, yielding the noisy reads $Z_{i,j,1},\dots,Z_{i,j,Q_{2,i,j}}$. 
Finally, the sequences $\left( Z_{i,j,k} \right)_{i \in [M], j \in [Q_{1,i}], k \in [Q_{2,i,j}]}$ are shuffled under a uniformly random permutation, yielding the channel output $\vec{Y}$.

As in prior works, we will assume that $L=\beta\log_{|\Sigma|} M$ for some constant $\beta>1$.
Indeed, it was shown in~\cite{heckel2017fundamental} that no positive rate is achievable when $\beta < 1$. Moreover, if we take $L$ to be super-logarithmic in $M$, then indexing each of the $M$ input sequences with a unique index incurs a vanishing rate loss, since each index only uses about $\log_{|\Sigma|} M$ symbols over $\Sigma$, and we recover the ordering of the output sequences without shuffling. %Henceforth, we take $L = \beta \log M$, where the parameter $\beta > 1$ is fixed and arbitrary.

\subsection{Basic notation for the channel model}

We write $\vec{Q} := ( \vec{Q}_i )_{i \in [M]}$, where $\vec{Q}_i := (Q_{1,i}, ( Q_{2,i,j} )_{j \in [Q_{1,i}]} )$ denotes the number of samples drawn at each stage for the $i$-th input sequence, and let $N_i := \sum_{j=1}^{Q_{1,i}} Q_{2,i,j}$ be the number of output sequences corresponding to noisy reads of $X_i$. We denote by $M_0 := |\{ i \in [M] : N_i = 0 \}|$ the number of input sequences for which no output is produced, and by $\qerase := \Pr(N_1 = 0)$ the probability that no output is produced for a given input sequence (note that $N_1,\dots,N_M$ are i.i.d.). 
Fixed values of $\vec{Q}$ are denoted by $\vec{q}$, and we write $m_0(\vec{q})$ to denote the value of $M_0$ when conditioning on $\vec{Q}=\vec{q}$.

\subsection{Capacity of synthesis--sequencing channels}

We now formally define the capacity of a synthesis--sequencing channel, which follows along standard lines.
To simplify the exposition we will sometimes omit the dependence of the synthesis--sequencing channel on the various parameters $(\beta,\Qsyn,\Qseq,\Chsyn,\Chseq)$.

Fix $R\in[0,1]$.
A code $\cC \subseteq \Sigma^{M\times L}$ of rate $R$ for the synthesis--sequencing channel corresponds to an encoder-decoder pair $(\Enc,\Dec)$ where $\Enc\colon[\lceil |\Sigma|^{MLR}\rceil]\to \cC$ is an injective map and $\Dec\colon(\Sigma^L)^*\to [\lceil |\Sigma|^{MLR}\rceil]$.
We may sometimes instead refer to the \emph{code} $\cC=\Enc([|\Sigma|^{MLR}])$, leaving its encoder and decoder implicit.
Note that each codeword of $\cC$ corresponds to $M$ sequences over $\Sigma$ of length $L$.
Suppose that a message $W$ is mapped to a codeword $\Enc(W)$ and sent through the synthesis--sequencing channel, yielding channel output $\vec{Y}$.
Denote the decoder output by $\widehat{W}=\Dec(\vec{Y})$.
Then, the average decoding error probability of $\cC$ is
\begin{equation*}
    P_e(\cC) = \frac{1}{\lceil |\Sigma|^{MLR} \rceil} \sum_{w=1}^{\lceil |\Sigma|^{MLR} \rceil} \Pr \left( \Dec(\vec{Y}) \neq w \mid W = w \right).
\end{equation*}
We are now ready to define achievable rates and the coding capacity.

\begin{definition}[Achievable rate and capacity]\label{def:achievability}
    A rate $R$ is said to be \emph{achievable} for the synthesis--sequencing channel if there exists a sequence of codes $(\cC_M)_{M\in\mathbb{N}}$ such that each $\cC_M\subseteq \Sigma^{M\times L}$ has rate $R$ and average decoding error probability $P_e(\cC_M)\to 0$ as $M\to\infty$ (recall that we assume that $L=\beta\log_{|\Sigma|} M$).
    The \emph{(coding) capacity} of the synthesis--sequencing channel is the supremum of the set of achievable rates.
\end{definition}

\section{Our main result and its consequences}\label{sec:main-result}
%\jnote{work in progress}

We focus on the synthesis--sequencing channel over the binary alphabet $\Sigma=\{0,1\}$, where the synthesis and sequencing discrete memoryless channels $\Chsyn$ and $\Chseq$ are binary symmetric channels with error probabilities $\psyn$ and $\pseq$, respectively, and determine its capacity under mild assumptions on the system's parameters. We denote the capacity of this channel by $\Css(\beta,\psyn,\pseq,\Qsyn,\Qseq)$ (see \Cref{sec:model} for a detailed discussion about the channel model).  The extension of our results to quaternary and other non-binary alphabets is natural. 

Informally, our result says that, under mild assumptions, $\Css$ equals the capacity of the corresponding channel \emph{without} shuffling, minus a $1/\beta$ penalty due to shuffling.
This extends prior results for sequencing-only channels to the synthesis--sequencing setting~\cite{shomorony2022information}.

To formally state our converse and achievability results, it is convenient to consider a simpler memoryless channel, which we call $\twobsc$. This channel differs from the synthesis--sequencing channel in two fundamental respects: (1) the number of samples drawn at each stage is fixed (deterministic) as a parameter of the channel, and (2) there is no shuffling of output sequences. Since the $\twobsc$ is memoryless, it is sufficient to consider single-bit inputs. Thus, for a fixed sampling vector $\vec{q} = \left(q_{1}, \left( q_{2,1} \dots q_{2,q_1} \right) \right)$ the channel $\twobsc(\psyn, \pseq, \vec{q})$ is identical to the synthesis--sequencing channel, where $M=L=1$, and where the sampling steps at the synthesis and sequencing stages are deterministic and specified by $\vec{q}$.
We further note that the $\twobsc$ is symmetric, and thus the \emph{channel entropy} $H(Y|X)$ is well defined (that is, $H(Y|X)$ is independent of the input $X$). We denote the capacity (resp. channel entropy) of $\twobsc(\psyn, \pseq, \vec{q})$ by $\Cbsc(\psyn, \pseq, \vec{q})$ (resp. $\Ebsc(\psyn, \pseq, \vec{q})$). For completeness, we derive an exact expression for $\Cbsc$ in \Cref{sec:cap_twobsc}.

We are now ready to state our converse and achievability results.

\begin{theorem}[Converse bound]\label{thm:converse}
     Assume that $\Qsyn$ and $\Qseq$ have bounded expectations. Define the \emph{effective error probability} $\peff := \psyn + \pseq - 2 \psyn \pseq$. Then, if $\peff < 1/4$ and $h(2 \peff) < 1 - 2/\beta$, with $h$ the binary entropy function, we have
    \begin{equation*}
        \Css(\beta,\psyn,\pseq,\Qsyn,\Qseq) \leq  \underset{\vec{q} \sim (\Qsyn,\Qseq^{\Qsyn})}{\E}\left[ \Cbsc(\psyn,\pseq, \vec{q}) \right] - \frac{1}{\beta}(1-\qerase).
    \end{equation*}
\end{theorem}

\begin{theorem}[Achievability bound]\label{thm:achievability}
     Assume that $\Qsyn$ and $\Qseq$ have bounded expectations. Then, if $\peff < 1/4, \pseq < \peff(1-\peff)$, and $h(2 \peff) < 1 - 2/\beta$, we have
    \begin{equation*}
        \Css(\beta,\psyn,\pseq,\Qsyn,\Qseq) \geq \underset{\vec{q} \sim (\Qsyn,\Qseq^{\Qsyn})}{\E}\left[ \Cbsc(\psyn,\pseq, \vec{q}) \right] - \frac{1}{\beta}(1-\qerase).
    \end{equation*}
\end{theorem}

\Cref{thm:converse} and \Cref{thm:achievability} are established in \Cref{sec:converse,sec:achievability}, respectively.
In the special case where $\pseq= 0$ and $\Pr(\Qseq = 1) = 1$, we recover the main result of~\cite{lenz2022noisy}.
%In the special case where $\psyn = 0$ and $\Pr(\Qsyn = 1) = 1$, similar results were obtained in ~\cite[Lemma 6]{lenz2022noisy} and ~\cite[Lemma 11]{lenz2022noisy} though with more restrictive assumptions on the sampling distribution. \snote{this is not true for the case when $\psyn=0$, as our condition $\pseq < \peff(1-\peff)$ implies that our result does not capture this setting}\knote{According to me, our results capture it in the sense that if $\Pr(\Qsyn = 1) = 1$, then the sequences in the subclustering step are indeed independent and thus we don't need $\pseq < \peff(1-\peff)$ condition anymore. This is what we were discussing in our meeting I think} \snote{perhaps we should mention that this condition is needed for the non-degenerate case in which our channel is not equivalent to the sequencing-only channel, and that if $\Qsyn=1$ or $\psyn=0$, the result holds without needing this condition, as already shown in \cite{lenz2022noisy}}.
We now discuss a concrete example illustrating new trade-offs stemming from our main result and a comparison to previous models.

Consider a sequencing-only channel with $M$ input sequences, each of length $L=100\log_2 M$, where each sequence is sampled independently according to a Poisson distribution $Q\sim \text{Poisson}(\lambda=6)$, and the output reads are obtained by transmitting each sampled sequence independently through a noisy memoryless channel, followed by random shuffling.
For simplicity, suppose the input sequences are binary and the noisy channel is a binary symmetric channel~(BSC) with error probability $p=0.2$. The capacity of this sequencing-only channel is around $0.77$, following the results of~\cite{lenz2022noisy}.

Now consider the synthesis--sequencing channel with the same input, with Poisson sampling and BSC errors at both stages, namely $\Qsyn\sim \mathrm{Poisson}(\lambdasyn)$ and \mbox{$\Qseq\sim \mathrm{Poisson}(\lambdaseq)$}, with BSC error probabilities $\psyn$ and $\pseq$. Fig.~\ref{fig:capacity} shows capacity results for the synthesis--sequencing channel under multiple settings for which the average number of reads at the output and the effective error rate are fixed to match those of the sequencing-only channel. More specifically, $\lambdasyn \lambdaseq=6$, and \mbox{$\peff=0.2$}. The capacity results reported in Fig.~\ref{fig:capacity} are computed from the converse and achievability bounds in \Cref{thm:converse} and \Cref{thm:achievability}. Fig.~\ref{fig:coverage} shows the coverage distributions at the output of both channels, i.e., the probability mass function of the number of reads obtained per input sequence.
% The coverage distributions of both channel outputs, showing the probability mass function of the number of reads obtained per input sequence, are illustrated in Fig.~\ref{fig:coverage}, for the particular case of $(\lambda_{\text{syn}},\lambda_{\text{seq}})=(6,1)$ for the synthesis--sequencing channel.

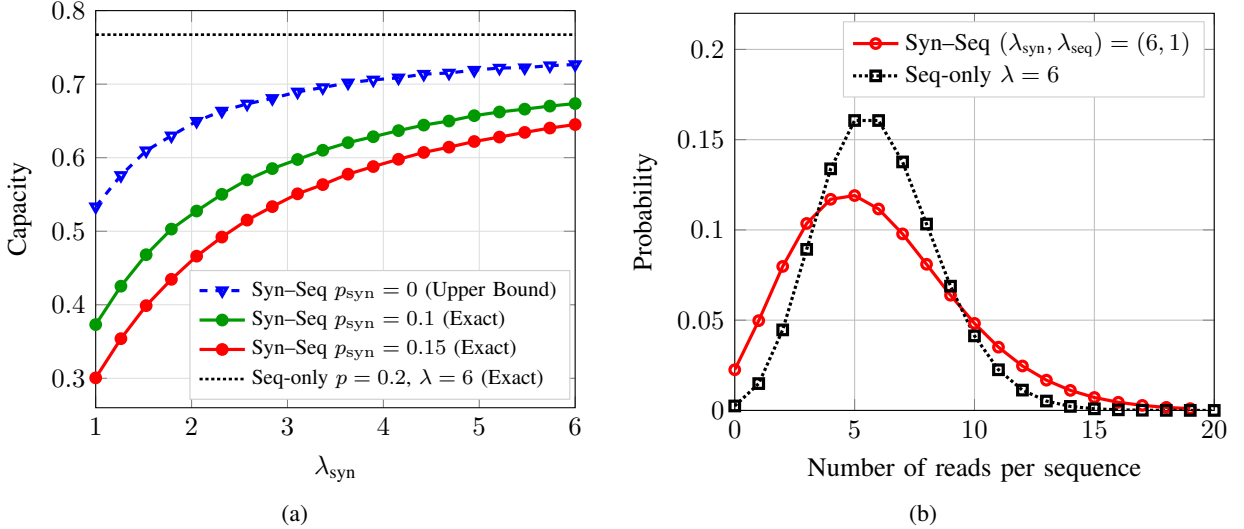
\begin{figure}[h!]
\centering

% ================= LEFT FIGURE =================
\subfloat[%
\label{fig:capacity}]
{
\begin{minipage}[b]{0.48\textwidth}
\centering
\begin{tikzpicture}
\begin{axis}[
width=\linewidth,
    font=\small,
    xmin=1, xmax=6,
    ymin=0.26, ymax=0.8,
    xtick={1,2,3,4,5,6},
    ytick={0.3,0.4,0.5,0.6,0.7,0.8},
    grid=major,
    grid style={gray!25},
    xlabel={$\lambda_{\text{syn}}$},
    ylabel={Capacity},
    legend style={
        at={(0.99,0.02)},
        anchor=south east,
        draw=gray!40,
        fill=white,
        font=\scriptsize
    },
    legend cell align=left,
    every axis plot/.append style={line width=1.2pt},
]

% Syn--Seq p1 = 0
\addplot[blue, dashed, mark=triangle, mark options={rotate=180, solid}, mark size=2pt] coordinates {
(1.0000,0.5333260798) (1.2632,0.5759482426) (1.5263,0.6095233146)
(1.7895,0.6297932841) (2.0526,0.6496280743) (2.3158,0.6634693781)
(2.5789,0.6731439333) (2.8421,0.6806798159) (3.1053,0.6894815875)
(3.3684,0.6954098527) (3.6316,0.7018175887) (3.8947,0.7058070824)
(4.1579,0.7086120531) (4.4211,0.7136114917) (4.6842,0.7151774737)
(4.9474,0.7190490961) (5.2105,0.7219172531) (5.4737,0.7223762746)
(5.7368,0.7251533458) (6.0000,0.7266950704)};
\addlegendentry{Syn--Seq $\psyn=0$ (Upper Bound)}

% % Syn--Seq p1 = 0.05
% \addplot[orange,mark=*,mark size=1.8pt] coordinates {
% (1.0000,0.4454902694) (1.2632,0.4967837323) (1.5263,0.5360926299)
% (1.7895,0.5640161186) (2.0526,0.5848920230) (2.3158,0.6045236546)
% (2.5789,0.6178817222) (2.8421,0.6313227984) (3.1053,0.6428062589)
% (3.3684,0.6525174222) (3.6316,0.6596098501) (3.8947,0.6663702727)
% (4.1579,0.6720353657) (4.4211,0.6791573857) (4.6842,0.6830045510)
% (4.9474,0.6868298523) (5.2105,0.6922204707) (5.4737,0.6931932462)
% (5.7368,0.6962528958) (6.0000,0.6994528875)};
% \addlegendentry{Syn--Seq $\psyn=0.05$}

% Syn--Seq p1 = 0.1
\addplot[green!60!black,mark=*,mark size=1.8pt] coordinates {
(1.0000,0.3730684544) (1.2632,0.4252837393) (1.5263,0.4680932917)
(1.7895,0.5029238248) (2.0526,0.5274799575) (2.3158,0.5501776668)
(2.5789,0.5698308698) (2.8421,0.5851765385) (3.1053,0.5975736898)
(3.3684,0.6100893787) (3.6316,0.6204785435) (3.8947,0.6285189272)
(4.1579,0.6368064481) (4.4211,0.6443186723) (4.6842,0.6498690585)
(4.9474,0.6571356400) (5.2105,0.6621548786) (5.4737,0.6660194918)
(5.7368,0.6701889796) (6.0000,0.6736041473)};
\addlegendentry{Syn--Seq $\psyn=0.1$ (Exact)}

% Syn--Seq p1 = 0.15
\addplot[red,mark=*,mark size=1.8pt] coordinates {
(1.0000,0.3006250149) (1.2632,0.3539126667) (1.5263,0.3987954185)
(1.7895,0.4345970424) (2.0526,0.4661129197) (2.3158,0.4920629164)
(2.5789,0.5151295595) (2.8421,0.5334723859) (3.1053,0.5508653400)
(3.3684,0.5633939270) (3.6316,0.5775826132) (3.8947,0.5879410157)
(4.1579,0.5978242621) (4.4211,0.6071481199) (4.6842,0.6143073808)
(4.9474,0.6220033077) (5.2105,0.6280200922) (5.4737,0.6343809801)
(5.7368,0.6403088335) (6.0000,0.6450547918)};
\addlegendentry{Syn--Seq $\psyn=0.15$ (Exact)}

\addplot[black, densely dotted, line width=1pt,domain=1:6]{0.7673037160826502};
\addlegendentry{Seq-only $p=0.2$, $\lambda=6$ (Exact)}

\end{axis}
\end{tikzpicture}
\end{minipage}
}
\hfill
% ================= RIGHT FIGURE =================
\subfloat[%
\label{fig:coverage}]
{
\begin{minipage}[b]{0.48\textwidth}
\centering
\begin{tikzpicture}
\begin{axis}[
    width=\linewidth,
    font=\small,
    yticklabel style={
        /pgf/number format/fixed,
        /pgf/number format/precision=2
    },
    xlabel={Number of reads per sequence},
    ylabel={Probability},
    ytick={0,0.05,0.1,0.15,0.2},
    ymin=0,
    xmin=0,
    ymax=0.22,
    xmax=20,
    grid=major,
    legend style={
        at={(0.98,0.98)},
        anchor=north east,
        draw=gray!40,
        fill=white,
        font=\footnotesize
    },
    legend cell align=left
]

\addplot[red,mark=o,mark size=1.8pt,line width=1.2pt] coordinates {
(0,0.0225341498)(1,0.0497391025)(2,0.0797635310)(3,0.1035725630)
(4,0.1169643270)(5,0.1190372950)(6,0.1115933120)(7,0.0977962724)
(8,0.0809673384)(9,0.0638289580)(10,0.0482044676)(11,0.0350441878)
(12,0.0246210341)(13,0.0167712347)(14,0.0111063828)(15,0.0071669168)
(16,0.0045155067)(17,0.0027825525)(18,0.0016795697)(19,0.0009943704)};
\addlegendentry{Syn--Seq $(\lambda_{\text{syn}},\lambda_{\text{seq}})=(6,1)$}

\addplot[black,densely dotted,mark=square,mark size=1.8pt,line width=1.2pt,domain=0:20,samples=21,mark options={solid}]
{exp(-6)*(6^x)/factorial(x)};
\addlegendentry{Seq-only $\lambda=6$}

\end{axis}
\end{tikzpicture}
\end{minipage}
}

\caption{Capacity (a) and coverage distributions (b) for the sequencing-only channel and several instances of the synthesis--sequencing channel. All capacity results are reported under a fixed average number of reads $\lambda_{\text{syn}}\lambda_{\text{seq}}=\lambda=6$ and a fixed effective error rate $\peff=\psyn+\pseq-2\psyn\pseq=p=0.2$. For $\psyn\in \{0.1,0.15\}$, the channel parameters satisfy the conditions of both \Cref{thm:converse} and \Cref{thm:achievability}, yielding matching bounds and hence the exact capacity. For $\psyn=0$, only the converse bound is shown since the conditions of \Cref{thm:achievability} do not hold.}
\label{fig:combined}
\end{figure}

Three key observations can be drawn from the numerical results in Fig.~\ref{fig:combined}. First, for fixed $\psyn\in \{0.1, 0.15\}$, the results show that the exact capacity is strongly influenced by the average number of synthesized strands per input sequence ($\lambda_{\text{syn}}$). While minimizing physical redundancy (reducing $\lambda_{\text{syn}}$) is desirable in practice to fully exploit the high storage density of DNA, these results indicate that this incurs a significant rate loss, which underscores a fundamental trade-off between information density (bits/nt) and storage density~(bits/gram). Second, for fixed $\lambda_{\text{syn}}$, the capacity decreases as the synthesis error rate increases from $\psyn=0.1$ to $\psyn=0.15$, and since the effective error rate $\peff=0.2$ is held constant across both settings, this rate loss reflects the impact of systematic, non-random errors across sequencing reads. Third, even in the case where $\psyn=0$ and $\lambda_{\text{syn}}=\lambda=6$, the upper bound (converse) shows that the capacity of the synthesis--sequencing channel is lower than that of the sequencing-only channel. This is a consequence of the more pronounced skew in the coverage distribution (Fig.~\ref{fig:coverage}), which is in line with the coverage bias profile observed in practice~\cite{gimpel2023digital}. Overall, these observations highlight key trade-offs across operating regimes of the DNA storage channel, providing further insight into how synthesis and sequencing jointly influence its fundamental limits.

\section{Converse bound}\label{sec:converse}

Inspired by \cite{lenz2022noisy}, we consider the alternative, equivalent definition of our channel in which the output sequences $(Z_{i,j,k} )_{i \in [M], j \in [Q_{1,i}], k \in [Q_{2,i,j}]}$ are first grouped together in $M$ clusters $\vec{Z} = ( \vec{Z}_i )_{i \in [M]}$, where cluster $\vec{Z}_i$ contains all output sequences $Z_{i,j,k}$ that come from $X_i$.
Each cluster $\vec{Z}_i$ has its elements grouped together in sub-clusters so that $\vec{Z}_i = ( \vec{Z}_{i,j} )_{j \in [Q_{1,i}]}$, where sub-cluster $\vec{Z}_{i,j} = (Z_{i,j,k})_{k \in [Q_{2,i,j}]}$ contains the output sequences $Z_{i,j,k}$ that come from the noisy synthesized sequence $Z_{i,j}$. %We further assume that the sub-clusters inside each cluster are ordered such that $|\vec{Z}_{i,j}| \geq |\vec{Z}_{i,k}|$ whenever $j < k$ (the ordering is otherwise arbitrary).
The channel then shuffles the ordering of the $M$ clusters under a random permutation $S:[M] \rightarrow [M]$, yielding clusters $\vec{Z}' = ( \vec{Z}_{i}' )_{i \in [M]}$, where $\vec{Z}_{i}' = \vec{Z}_{S(i)}$ (and $Z_{i,j,k}' = Z_{S(i),j,k}$). Finally, the channel outputs a shuffling $\vec{Y}$ of all individual output sequences under a random permutation $P$.

\begin{remark}
    Due to the final permutation, the cluster permutation $S$ is redundant. However, it is useful to define the channel in this way, since we restrict ourselves to analyzing the genie-aided channel which reveals the permuted clusters $( \vec{Z}_{i}' )_{i \in [M]}$ and the full sampling distribution $\vec{Q}' = ( \vec{Q}_i' )_{i \in [M]}$, where $\vec{Q}_i' = \vec{Q}_{S(i)}$.
\end{remark}

Our approach to proving the converse bound is an extension of that of \cite{lenz2022noisy}. As already mentioned, it will be sufficient to consider the simplified channel model in which the permuted clusters are revealed, as well as the number of samples produced at each stage for each sequence. As we will see, the entropy loss resulting from this process is asymptotically vanishing.

After decomposing the expression for mutual information as a combination of more tractable entropy terms, the main challenge comes from jointly bounding the terms in the expression $H(\vec{Z}'\mid \vecqp) + H(S\mid\vecqp,\vec{X},\vec{Z}')$, where $\vecqp$ characterizes the number of samples produced at each stage for each sequence. The reason why one can do better than naively bounding each term individually is due to the following observation, initially made in \cite{shomorony2019capacity}. Intuitively, if the output sequences $\vec{Z}'$ have high entropy (if they are likely ``spread out'' in a Hamming distance sense), then the conditional entropy of the permutation $S$ will be lower, since clusters can be identified more easily. Conversely, if the output sequences $\vec{Z}'$ tend to be close to each other, then the output entropy is lower but there remains more uncertainty about $S$. This dichotomy can be exploited in order to bound the aforementioned expression in a non-trivial way. The basic approach to formalizing this argument is to define a set $\cU$ of indices $i \in [M]$ for which the corresponding outputs are ``far apart'' from each other (in a Hamming sense), and to bound the two entropy terms as a function of the expected size of $\cU$. We remark that $H(\vec{Z}'\mid\vecqp)$ will then be increasing in the size of $\cU$, and $H(S\mid\vecqp,\vec{X},\vec{Z}')$ will be decreasing.

\subsection{Proof of converse bound}
Let $(\cC_M)$ be a sequence of codes with rate $R$ and vanishing decoding error probability $P_e(\cC_M)$. Take $\vec{X}$ to be a codeword chosen uniformly from $\cC_M$ and let $\vec{Y}$ be the channel output on $\vec{X}$. Using Fano's inequality, and the fact that $H(\vec{X}) = MLR$,
\begin{equation*}
    I(\vec{X};\vec{Y}) \geq MLR(1-P_e(\cC_M)) - 1,
\end{equation*}
which implies that
\begin{equation*}
    R \leq \frac{I(\vec{X};\vec{Y})}{ML} + o(1).
\end{equation*}
It then suffices to upper-bound the mutual information $I(\vec{X};\vec{Y})$ for all inputs $\vec{X} = \left(X_1,\dots,X_M\right) \in \{0,1\}^{ML}$ to the channel. We now verify that we can restrict ourselves to the genie-aided channel which outputs the clusters $\vec{Z}'$ and the number of samples $\vec{Q}'$ produced at each step for each sequence. Indeed,
\begin{align*}
    I(\vec{X};\vec{Y}) &\leq H(\vec{Y}) - H(\vec{Y}\mid\vec{X}) \\
    &\leq H(\vec{Z}',P,\vec{Q}') - H(\vec{Z}',P,\vec{Q}'\mid\vec{X}) \\
    &= H(P,\vec{Q}') - H(P,\vec{Q}'\mid\vec{X}) + H(\vec{Z}'\mid\vec{Q}',P) - H(\vec{Z}'\mid\vec{Q}',P,\vec{X}) \\
    &= H(\vec{Z}'\mid\vec{Q}') - H(\vec{Z}'\mid\vec{Q}',\vec{X}) \\
    &= I(\vec{X};\vec{Z}'\mid\vec{Q}').
\end{align*}
We can expand the expression for mutual information as
\begin{align}\label{eq:firstMIbound}
    I(\vec{X};\vec{Z}'\mid\vec{Q}') &= H(\vec{Z}'\mid\vec{Q}') - H(\vec{Z}'\mid\vec{Q}',\vec{X}) \notag \\
    &= H(\vec{Z}'\mid\vec{Q}') - H(\vec{Z}',S\mid\vec{Q}',\vec{X}) + H(S\mid\vec{Q}', \vec{X}, \vec{Z}') \notag \\
    &= H(\vec{Z}'\mid\vec{Q}') -H(S\mid\vec{Q}',\vec{X}) - H(\vec{Z}'\mid\vec{Q}',\vec{X,}S) + H(S\mid\vec{Q}', \vec{X}, \vec{Z}') \notag \\
    &= H(\vec{Z}'\mid\vec{Q}') - H(\vec{Z}'\mid\vec{Q}',\vec{X},S) + H(S\mid\vec{Q}', \vec{X}, \vec{Z}') - M \log M + O(M) \notag \\
    &= \sum_{\vecqp} \Pr(\vec{Q}'=\vecqp) \cdot \left[ H(\vec{Z}'\mid\vecqp) - H(\vec{Z}'\mid\vecqp,\vec{X},S) + H(S\mid\vecqp, \vec{X}, \vec{Z}') \right] - \frac{ML}{\beta} + o(ML),
\end{align}
where we used the chain rule for entropy, the fact that $H(S\mid\vec{Q}',\vec{X}) = \log M! = M \log M - O(M)$, and the fact that $L = \beta \log M$.
Thus, it is sufficient to jointly bound the three entropy terms in \Cref{eq:firstMIbound}. We state these bounds in the claims below, which are proven subsequently.
We start with the simplest case.
\begin{claim} \label{claim:entropy_bound1}
For each $\vecqp \in \supp(\vec{Q}')$,
    \begin{equation*}
        H(\vec{Z}'\mid\vecqp,\vec{X},S) = L \cdot \sum_{i=1}^{M} \Ebsc(\psyn,\pseq,\vecqpi).
    \end{equation*}
\end{claim}
We will need a better upper-bound for the channel output term $H(\vec{Z}'\mid\vecqp)$ than the trivial 
\begin{equation*}
    H(\vec{Z}'\mid\vecqp) \leq L \cdot \sum_{i=1}^{M} \left(\Cbsc(\psyn,\pseq,\vec{q}_i) + \Ebsc(\psyn,\pseq,\vec{q}_i)\right).
\end{equation*}
To this end, we define a set $\cU$ such that the output entropy for indices not in $\cU$ can be bounded more tightly.
\begin{definition} \label{def:U}
    Fix a radius $\tau \in [0,1]$. We define $\cU$ as the largest set of indices $i \in [M]$ for which the following two conditions hold:
    \begin{enumerate}
        \item $n_i' := \sum_{j=1}^{q_{1,i}'} q_{2,i,j}' > 0$.

        \item If $j \neq i$ and $j \in \cU$, then $d_H(Z_{i,1,1}',Z_{j,1,1}') > \tau L$.
    \end{enumerate}
\end{definition}
\begin{remark}
    If $n_i' > 0$, it is without loss of generality to assume that $Z_{i,1,1}'$ exists (i.e., that an output is produced for the first synthesized sequence $Z_{i,1}'$ coming from $X_{S(i)}$).
    %Even if $q_i' > 0$, it is possible that no output is produced for the synthesized sequence $Z_{i,1}'$, in which case $Z_{i,1,1}'$ may not exist. However, there must be some output $Z_{i,j,1}'$ produced for some synthesized sequence $Z_{i,j}'$. Let $j_1$ be the smallest index such that $Z_{i,j_1,1}'$ exists. For simplicity, in \Cref{def:U} we abuse notation and write $Z_{i,1,1}'$ to denote $Z_{i,j_1,1}'$.
\end{remark}

\begin{notation}
    We shall use the notation $U_{\vecqp}$ to denote $\E[|\cU| \mid \vecqp]$. For a set $u \subseteq [M]$, we write $\vec{Z}_u'$ to denote the restriction of $\vec{Z}'$ to indices belonging to $u$.
\end{notation}

\begin{claim} \label{claim:entropy_bound2}
    For each $\vecqp \in \supp(\vec{Q}')$ and $0 \leq \tau \leq 1/2$, for any $\eps > 0$ and large enough $L$, it holds that
    \begin{multline*}
        H(\vec{Z}'\mid\vecqp) \leq L \cdot \sum_{i: n_i' > 0} (\Ebsc(\psyn,\pseq,\vecqpi) + \Cbsc(\psyn,\pseq,\vecqpi) + \eps \cdot n_i') \\
        + (M-m_0(\vecqp)-U_{\vecqp}) \cdot (\log U_{\vecqp} + L(h(\tau)-1) + \eps L) + o(ML).
    \end{multline*}
\end{claim}

We now bound the third entropy term in \Cref{eq:firstMIbound}.

\begin{claim} \label{claim:entropy_bound3}
    For each $\vecqp \in \supp(\vec{Q})$, and for $\tau > 2 \peff$,
    \begin{align*}
        H(S\mid\vec{X},\vec{Z}',\vecqp) \leq M \log M - U_{\vecqp} \log U_{\vecqp} + o(ML).
    \end{align*}
\end{claim}

Plugging \Cref{claim:entropy_bound1,claim:entropy_bound2,claim:entropy_bound3} into \Cref{eq:firstMIbound}, we have that $I(\vec{X};\vec{Z}'\mid\vec{Q}')$ is at most the expectation over $\vecqp \sim \vec{Q}'$ of
\begin{align*}
    &-\frac{ML}{\beta} + M \log M - U_{\vecqp} \log U_{\vecqp} \\
    &+ L \cdot \sum_{i: n_i' > 0} (\Cbsc(\psyn,\pseq,\vecqpi) + \eps \cdot n_i') \\
        &+ (M-m_0(\vecqp)-U_{\vecqp}) \cdot \left( \log U_{\vecqp} + L(h(\tau)-1) + \eps L \right) + o(ML).
\end{align*}
For fixed $\vecqp$, define
\begin{align*}
    f(U_{\vecqp}) &:= - U_{\vecqp} \log U_{\vecqp} + (M-m_0(\vecqp) - U_{\vecqp}) \cdot (\log U_{\vecqp} + L(h(\tau)-1 + \eps))
\end{align*}
to be the terms that depend on $U_{\vecqp}$. The derivative of $f$ satisfies
\begin{align*}
    f'(U_{\vecqp}) &= -\log e - \log U_{\vecqp} + \log e \cdot \frac{M - U_{\vecqp} - m_0(\vecqp)}{U_{\vecqp}} - \log U_{\vecqp} + L(1-h(\tau)-\eps) \\
    &> L(1-h(\tau)-\eps) - 2 \log U_{\vecqp} - \log e.
\end{align*}
Thus, $f'(U_{\vecqp}) > 0$ whenever
\begin{equation*}
    U_{\vecqp} < \left( \frac{M^{(1-h(\tau)-\eps) \cdot \beta}}{e} \right)^{\frac{1}{2}}.
\end{equation*}
The exponent of $M$ is larger than $1$ as long as $1-h(\tau)-\eps > 2/\beta$, in which case, for $M$ large enough, $f$ is increasing for all $0 \leq U_{\vecqp} \leq M-m_0(\vecqp)$. Note also that, since we can take $\eps$ to be arbitrarily small by taking $M$ large enough, it is sufficient to guarantee that $1-h(\tau) > 2/\beta$. Thus, $f$ attains its maximum at
\begin{align*}
    f(M-m_0(\vecqp)) &= -(M-m_0(\vecqp)) \log (M-m_0(\vecqp)).
\end{align*}
Note that $\E_{\vec{Q}'}[m_0(\vecqp)] = M \cdot \qerase$, where the expectation is taken over $\vecqp \sim Q'$. Thus,
\begin{align*}
    \E_{\vec{Q}'} \left[ -(M-m_0(\vecqp)) \log (M-m_0(\vecqp)) \right] &\overset{(a)}{\leq} -M(1-\qerase) \log (M(1-\qerase)) \\
    &= -M(1-\qerase) \log M + O(M) \\
    &= -\frac{ML}{\beta}(1-\qerase) + o(ML),
\end{align*}
where in inequality $(a)$ we applied Jensen's inequality to the concave function $x \to -x \log x$, taking $x = M-m_0(\vecqp)$.
It therefore holds that
\begin{align*}
    \E_{\vec{Q}'} \left[ M \log M -(M-m_0(\vecqp)) \log (M-m_0(\vecqp)) \right] \leq \frac{ML}{\beta} \cdot \qerase + o(ML).
\end{align*}
Finally, using linearity of expectation and the fact that $\vec{Q}'$ is a product distribution, we obtain the capacity upper bound expression
\begin{equation*}
    \underset{ \vec{q} \sim (\Qsyn,\Qseq^{\Qsyn})}{\E}\left[ \Cbsc(\psyn,\pseq, \vec{q}) + \eps \cdot \sum_{j=1}^{q_1} q_{2,j} \right] - \frac{1}{\beta}(1- \qerase).
\end{equation*}
Given that $\Qsyn,\Qseq$ have bounded expectations, the sum $\sum_{j=1}^{q_1} q_{2,j}$ above has a bounded expectation by Wald's equation. Since $\eps$ can be taken arbitrarily small, the proof is concluded.

\subsection{Proofs of intermediate claims}

It remains to provide the proofs of the claims used to show the converse bound.

\begin{proof}[Proof of \Cref{claim:entropy_bound1}]
    Since $S$ defines a bijection between $\vec{Z}$ and $\vec{Z}'$ (and between $\vec{q}$ and $\vecqp$), we have
    \begin{equation*}
        H(\vec{Z}'\mid\vecqp,\vec{X},S) = H(\vec{Z}\mid\vec{q},\vec{X},S) = \sum_s \Pr(S=s) \cdot H(\vec{Z}\mid\vec{q},\vec{X},s).
    \end{equation*}
    Moreover, we have
    \begin{equation*}
        H(\vec{Z}\mid\vec{q},\vec{X},s) = \sum_{i=1}^{M} H(\vec{Z}_i\mid\vec{q},\vec{X},s) = \sum_{i=1}^{M} H(\vec{Z}_i\mid\vec{q},X_i,s),
    \end{equation*}
    since, given $\vec{q}$, $s$ and $X_i$, the output cluster $\vec{Z}_i$ is independent of other inputs and corresponding output clusters.

    Since we have fixed the number of samples at each step, the output cluster $\vec{Z}_i$ is the output of the $\twobsc$ on input $X_i$. Thus,
    \begin{equation*}
        H(\vec{Z}_i\mid\vec{q},X_i,s) = L \cdot \Ebsc(\psyn,\pseq,\vec{q}_i),
    \end{equation*}
    and the result follows.
\end{proof}

We shall use the following lemma to prove \Cref{claim:entropy_bound2}.

\begin{lemma}\label{lemma:typical_errors}
        Fix parameters $\psyn,\pseq,\vecqpi$ of the $\twobsc$ channel. Let $E_{i,j,k} := X_{S(i)} + Z_{i,j,k}'$, where addition is componentwise and mod 2, be the error associated with $Z_{i,j,k}'$, and let $E_{i,j,k}'$ denote the first bit of the length-$L$ sequence $E_{i,j,k}$. Define
        \begin{equation*}
            \vec{E}_i := \left( E_{i,j,k} + E_{i,1,1}\right)_{j \in [q_{1,i}'], k \in [1:q_{2,i,j}'],(j,k)\neq(1,1)}
        \end{equation*}
        and
        \begin{equation*}
            \vec{E}_i' := \left( E_{i,j,k}' + E_{i,1,1}'\right)_{j \in [q_{1,i}'], k \in [1:q_{2,i,j}'],(j,k)\neq(1,1)}.
        \end{equation*}
        %and let $\vec{E}_i'$ be distributed as the restriction of $\vec{E}_i$ to the first bit of each subsequence $(E_{i,j,k} + E_{i,1,1})$ of length $L$.
        For $\eps > 0$, let $\cF$ be the event that $\vec{E}_i$ is $\eps$-typical, where we define $\eps$-typical sequences as those that belong to the set
\begin{equation*}
    A_{\eps}^L := \left\{ \vec{e}_i: \left| -\frac{\log \Pr(\vec{E}_i = \vec{e}_i)}{L} - H(\vec{E}_i') \right| < \eps \right\}.
\end{equation*} Then,
        \begin{enumerate}
            \item There exists $L_{\vecqpi}(\eps)$ such that $\Pr(\cF) \geq 1-\eps$ for every $L \geq L_{\vecqpi}(\eps)$.

            \item There are at most $2^{L(H(\vec{E}_i') + \eps)}$ $\eps$-typical sequences.

            \item The entropy of $\vec{E}_i'$ satisfies
            \begin{equation*}
                H(\vec{E}_i') \leq \Ebsc(\psyn,\pseq,\vecqpi) + \Cbsc(\psyn,\pseq,\vecqpi) - 1
            \end{equation*}
        \end{enumerate}
\end{lemma}
\begin{proof}
    The first two items follow from the standard asymptotic equipartition property (see, e.g., \cite[Chapter 3]{Cover2006}).
We first show that
\begin{equation*}
    \vec{E}_i' \indep W := X + E_{i,1,1}',
\end{equation*}
where $X$ is the uniform input distribution over $\{0,1\}$. Indeed,
\begin{align*}
    \Pr(\vec{E}_i' = \vecep \mid W = w) &= \Pr(\vec{E}_{i}' = \vecep, W = w) / \Pr(W = w) \\
    &= 2 \cdot \Pr(\vec{E}_{i}' = \vecep, W = w) \\
    &= 2 \cdot \sum_{e} \Pr(E_{i,1,1}' = e, \vec{E}_{i}' = \vecep, W = w) \\
    &= 2 \cdot \sum_{e} \Pr(X = w - e) \cdot \Pr(E_{i,1,1}' = e, \vec{E}_{i}' = \vecep) \\
    &= \sum_{e} \Pr(E_{i,1,1}' = e, \vec{E}_{i}' = \vecep) \\
    &= \Pr(\vec{E}_{i}' = \vecep).
\end{align*}
Thus, we have
\begin{align*}
    H(\vec{E}_{i}') &= H(\vec{E}_{i}'\mid W) \\
    &\overset{(a)}{=} H((E_{i,j,k}' + X)_{j \in [q_{1,i}'],k \in [1:q_{2,i,j}'],(j,k)\neq(1,1)}\mid W) \\
    &= H((E_{i,j,k}' + X)_{j \in [q_{1,i}'],k \in [1:q_{2,i,j}']}) - H(W) \\
    &\overset{(b)}{=} \Cbsc(\psyn,\pseq,\vecqpi) + \Ebsc(\psyn,\pseq,\vecqpi) - 1,
\end{align*}
where equality $(a)$ follows from the bijection obtained by conditioning on $W$, and in equality $(b)$ we used the fact that the uniform distribution is capacity-achieving, and so the entropy of the channel output corresponds to the capacity plus the channel entropy.
\end{proof}

\begin{proof}[Proof of \Cref{claim:entropy_bound2}]
    It is easy to see that
    \begin{align} \label{eq:condition_on_u}
        H(\vec{Z}'\mid\vecqp) &\leq H(\vec{Z}', \cU \mid \vecqp) \notag \\
        &\leq H(\cU\mid\vecqp) + H(\vec{Z}'\mid\vecqp,\cU) \notag \\
        &\leq M + \sum_{u \subseteq [M]} \Pr(\cU = u\mid\vecqp) \cdot H(\vec{Z}'\mid\vecqp,u).
    \end{align}
    Now, for fixed $u$, we can split
    \begin{equation} \label{eq:split_u_uc}
        H(\vec{Z}'\mid\vecqp,u) \leq H(\vec{Z}_{u}'\mid\vecqp,u) + H(\vec{Z}_{u^{c}}'\mid\vecqp,u,\vec{Z}_u').
    \end{equation}
    We can trivially bound the first term as
    \begin{align} \label{eq:trivial_bound_u}
        H(\vec{Z}_{u}'\mid\vecqp,u) \leq \sum_{i \in u} H(\vec{Z}_{i}'\mid\vecqpi,u)
        \leq L \cdot \sum_{i \in u} \left( \Cbsc(\psyn,\pseq,\vecqpi) + \Ebsc(\psyn,\pseq,\vecqpi) \right).
    \end{align}
    For the second term, we have
    \begin{equation} \label{eq:expand_uc}
        H(\vec{Z}_{u^{c}}'\mid\vecqp,u,\vec{Z}_u') \leq \sum_{i \in u^c: n_i' > 0} H(\vec{Z}_{i}'\mid\vecqpi,u,\vec{Z}_u'),
    \end{equation}
    since the entropy is zero when $n_i' = 0$.

    For fixed $i \in u^c$ and $\eps > 0$, we let $F_i$ be the indicator random variable for the event that the sequence $\vec{E}_i$ is $\eps$-typical (as defined in \Cref{lemma:typical_errors}). Then,
    \begin{align} \label{eq:condition_on_Fi}
        H(\vec{Z}_{i}'\mid\vecqpi,u,\vec{Z}_u') &\leq H(\vec{Z}_{i}',F_i\mid\vecqpi,u,\vec{Z}_u') \notag \\
        &\overset{(a)}{\leq} 1 + \Pr(F_i=0\mid\vecqpi,u,\vec{Z}_u') \cdot n_i' L + H(\vec{Z}_{i}'\mid\vecqpi,u,\vec{Z}_u',F_i=1),
    \end{align}
    where in inequality $(a)$ we used the trivial bounds $H(F_i\mid\vecqpi,u,\vec{Z}_u') \leq 1$ and $H(\vec{Z}_{i}'\mid\vecqpi,u,\vec{Z}_u',F_i=0) \leq n_i' L$.

    Note that, as long as $\tau \leq \frac{1}{2}$,
    \begin{align} \label{eq:bound_Zi11}
        H((\vec{Z}_{i}'\mid\vecqpi,u,\vec{Z}_u',F_i=1) &\leq H((\vec{Z}_{i,j,k}')_{(j,k) \neq (1,1)}\mid\vecqpi,u,\vec{Z}_u',F_i=1,Z_{i,1,1}') \notag \\
        &+ H(Z_{i,1,1}'\mid\vecqpi,u,\vec{Z}_u',F_i=1) \notag \\
        &\overset{(a)}{\leq} H(\vec{E}_i\mid\vecqpi,u,\vec{Z}_u',F_i=1,Z_{i,1,1}') + \log |u| + Lh(\tau),
    \end{align}
    where inequality $(a)$ follows from the fact that (1)
    since $i \in u^c$, given $u$ and $\vec{Z}_u'$, there are at most $|u| \cdot |B_H(L,\tau)|$ possibilities for $Z_{i,1,1}'$, where $|B_H(L,\tau)| \leq 2^{Lh(\tau)}$, and (2) given $\vecqpi$ and $Z_{i,1,1}'$, there exists a bijection between $(\vec{Z}_{i,j,k}')_{(j,k) \neq (1,1)}$ and $\vec{E}_i'$.

    Using \Cref{eq:bound_Zi11} and \Cref{lemma:typical_errors}, if we condition on $F_i = 1$, we have
    \begin{equation} \label{eq:Fi=1_entropy}
        H(\vec{Z}_{i}'\mid\vecqpi,u,\vec{Z}_u',F_i=1) \leq L \cdot \left(\eps + \Ebsc(\psyn,\pseq,\vecqpi) + \Cbsc(\psyn,\pseq,\vecqpi) - 1 \right) \\
        + \log |u| + Lh(\tau).
    \end{equation}
    Moreover, for $L \geq L_{\vecqpi}(\eps)$, we have $\Pr(F_i = 0) \leq \eps$.
    Thus, substituting \Cref{eq:Fi=1_entropy} in \Cref{eq:condition_on_Fi}, for $i \in u^c$ we have
    \begin{multline} \label{eq:uc_bound}
        H(\vec{Z}_{i}'\mid\vecqpi,u,\vec{Z}_u') \leq L \cdot (\Ebsc(\psyn,\pseq,\vecqpi) + \Cbsc(\psyn,\pseq,\vecqpi)) \\
        + \log |u| + L(h(\tau)-1) + 1 + \eps L \cdot (1 + n_i').
    \end{multline}
    It then holds that
    \begin{align} \label{eq:HZ-given-u}
        H(\vec{Z}'\mid\vecqp,u) &\overset{(a)}{\leq} L \cdot \sum_{i: n_{i}' > 0} (\Ebsc(\psyn,\pseq,\vecqpi) + \Cbsc(\psyn,\pseq,\vecqpi)) \notag \\
        &+ \sum_{i \in u^c: n_{i}' > 0} \left( \log |u| + L(h(\tau)-1) + 1 + \eps L \cdot (1 + n_i') \right) \notag \\
        &\overset{(b)}{=} L \cdot \sum_{i: n_i' > 0} (\Ebsc(\psyn,\pseq,\vecqpi) + \Cbsc(\psyn,\pseq,\vecqpi) + \eps \cdot n_i') \notag \\
        &+ (M-m_0(\vecqp)-|u|) \cdot \left( \log |u| + L(h(\tau)-1) + 1 + \eps L \right),
    \end{align}
    where inequality $(a)$ follows from simultaneously plugging \Cref{eq:trivial_bound_u,eq:expand_uc} into \Cref{eq:split_u_uc} and \Cref{eq:uc_bound} into \Cref{eq:expand_uc}, and equality $(b)$ follows from rearranging terms and noting that $\left| \{i \in u^c : n_i' > 0\} \right| = M - m_0(\vecqp) - |u|$.
    We then obtain
\begin{multline*}
    H(\vec{Z}' \mid \vecqp)
        \leq L \cdot \sum_{i: n_i' > 0} (\Ebsc(\psyn,\pseq,\vecqpi) + \Cbsc(\psyn,\pseq,\vecqpi) + \eps \cdot n_i') \\
        + (M-m_0(\vecqp)-U_{\vecqp}) \cdot \left( \log U_{\vecqp} + L(h(\tau)-1) + \eps L \right) + o(ML)
\end{multline*}
by inserting \Cref{eq:HZ-given-u} into \Cref{eq:condition_on_u} and applying Jensen's inequality to the concave function $|u|\mapsto -|u|\log|u|$.
\end{proof}

\begin{proof}[Proof of \Cref{claim:entropy_bound3}]
    We let $S_i$ denote $S(i)$. Since $\cU$ is a function of $\vec{Z}'$, we have
    \begin{align} \label{eq:sum_over_Si}
        H(S\mid\vec{X},\vec{Z}',\vecqp) \leq H(S\mid\vec{X},\vec{Z}',\vecqp,\cU) \leq \sum_u \Pr(\cU = u \mid \vecqp) \cdot \sum_{i=1}^{M} H(S_i\mid\vec{X},\vec{Z}',\vecqp,u).
    \end{align}
    If $n_i' = 0$, we trivially bound 
    \begin{equation} \label{eq:Si_when_qi=0}
        H(S_i\mid\vec{X},\vec{Z}',\vecqp,u) \leq \log M.
    \end{equation}
    For fixed $\delta > \peff$, we define $E_i$ to be the indicator random variable for the event that $n_i' > 0$ and $d_H(X_{S_i}, Z_{i,1,1}') > \delta L$.
    Since $d_H(X_{S_i}, Z_{i,1,1}') \sim \Bin(L,\peff)$, by Hoeffing's inequality we have
    \begin{equation*}
        \Pr(E_i = 1) \leq e^{-2L(\delta-\peff)^2}.
    \end{equation*}
    Now, we split
    \begin{align} \label{eq:condition_on_Ei}
        H(S_i\mid\vec{X},\vec{Z}',\vecqp,u) &\leq H(S_i,E_i\mid\vec{X},\vec{Z}',\vecqp,u) \notag \\
        &\overset{(a)}{\leq} 1 + \Pr(E_i=1) \log M + H(S_i\mid\vec{X},\vec{Z}',\vecqp,u,E_i=0),
    \end{align}
    where in inequality $(a)$ we used the facts that $H(E_i\mid\vec{X},\vec{Z}',\vecqp,u) \leq 1$, $H(S_i\mid\vec{X},\vec{Z}',\vecqp,u,E_i=1) \leq \log M$, and $\Pr(E_i = 0) \leq 1$.
    It remains to bound the last term from above. Set $\delta := \tau/2$ (this is possible if $\tau > 2 \peff$). For $i$ such that $n_i' > 0$, define
    \begin{equation*}
        A_i := \{i' \in [M] : d_H(X_{i'}, Z_{i,1,1}') \leq \delta L \}.
    \end{equation*}
    Clearly, $S_i \in A_i$ if we condition on $E_i=0$. Moreover, if $i,j \in u$ then $A_i \cap A_j = \emptyset$. To see this, note that $i' \in A_i$ implies that
    \begin{equation*}
        d_H(X_{i'}, Z_{j,1,1}') \overset{(a)}{\geq} d_H(Z_{i,1,1}', Z_{j,1,1}') - d_H(X_{i'}, Z_{i,1,1}') > \tau L - \delta L = \delta L,
    \end{equation*}
    where we used the triangle inequality for the Hamming distance in inequality $(a)$.
    Then,
    \begin{align} \label{eq:Si_when_Ei=0}
        \sum_{i:n_i'>0} H(S_i\mid\vec{X},\vec{Z}',\vecqp,u,E_i = 0) &\leq \sum_{i \in u^c : n_i' > 0} \log M + \sum_{i \in u} \log |A_i| \notag \\
        &\overset{(a)}{\leq} (M - m_0(\vecqp) - |u|) \cdot \log M + |u| \log (M/|u|) \notag \\
        &= (M - m_0(\vecqp)) \log M - |u| \log |u|,
    \end{align}
    where in inequality $(a)$ we used Jensen's inequality on the logarithm function, and used the disjointness of the $A_i$ to bound $\sum_{i \in u} |A_i| \leq M$.

    We conclude that
    \begin{align*}
        H(S\mid\vec{X},\vec{Z}',\vecqp) &\overset{(a)}\leq \sum_{u} \Pr(\cU = u\mid\vecqp) \cdot \left( M \log M - |u| \log |u| + (M - m_0(\vecqp)) \cdot (1 + e^{-2L(\delta-\peff)^2} \log M) \right) \\
        &\overset{(b)}{\leq} M \log M - U_{\vecqp} \log U_{\vecqp} + o(ML),
    \end{align*}
    where inequality $(a)$ follows from \Cref{eq:sum_over_Si,eq:Si_when_qi=0,eq:condition_on_Ei,eq:Si_when_Ei=0}, and inequality $(b)$ follows from Jensen's inequality applied to the concave function $|u|\mapsto -|u|\log|u|$.
\end{proof}

%\newpage 

\section{Achievability bound}\label{sec:achievability}

\begin{algorithm}[!ht] 
\caption{Clustering Algorithm} \label{clusteringalgo}
\begin{algorithmic}[1]
\Require $N$ received sequences $\{Z_{i,j,k}:{i \in [M], j \in [Q_{1,i}], k \in [Q_{2,i,j}]}\}$, cluster radius $\tau L$, sub-cluster radius $\phi L$
\Ensure $M$ clusters $(\hat{\vec{Z}}_i)_{i \in [M]}$ and sub-clusters $(\hat{\vec{Z}}_{i,j})_{j \in \hat{Q}_{1,i}} \forall i \in [M]$ 
\State $\cY \gets \{Z_{i,j,k}: i \in [M], j \in [Q_{1,i}], k \in [Q_{2,i,j}]\}$
\State $\hm \gets 0$
\While{$\cY \neq \emptyset$}
    \State $\hm \gets \hm + 1$
    \ForAll{$Z_{i,j,k} \in \mathcal{Y}$}
        \If{$|\hat{\vec{Z}}_{\hm}| = 0$ or  $(d_{H}(Z_{i,j,k},\hat{\vec{Z}}_{\hat{M}}(1)) < \tau L)$}
            \State $\vhz_{\hm} \gets Z_{i,j,k}$
             \State $\cY \gets \cY \setminus Z_{i,j,k}$
        \EndIf
    \EndFor
    \State $\cV = \hat{\vec{Z}}_{\hm}$
    \State $\hq_{1,\hm} \gets 0$
    \While{$\cV \neq \emptyset$}
    \State $\hq_{1,\hm} \gets \hq_{1,\hm} + 1$ 
    \ForAll{$Z_{i,j,k} \in \cV$}
    \If{$|\vhz_{\hm,\hq_{1,\hm}}| = 0$ or $(d_{H}(Z_{i,j,k},\vhz_{\hat{M},\hq_{1,\hm}}(1)) < \phi L)$}
    \State $\vhz_{\hm,\hq_{1,\hm}} \gets Z_{i,j,k}$
    \State $\cV \gets \cV \setminus Z_{i,j,k}$
    \EndIf
    \EndFor
    \EndWhile
\EndWhile
\If{$\hat{M} > M$}
    \State discard $\vhz_{M+1},\dots,\vhz_{\hat{M}}$
\ElsIf{$\hat{M} < M$}
    \State add empty clusters $\vhz_{\hat{M}+1},\dots,\vhz_M$
\EndIf
\end{algorithmic}
\end{algorithm}

We proceed by deriving achievable rates using standard random coding techniques. Firstly, we choose a random codebook $\cC$ of rate $R$ with i.i.d. codewords that are drawn from a given input distribution $\Pr(X)$. Then, we derive a suitable decoder and compute its average error probability, averaged over all codebooks. For a fixed rate $R$, we show that the average decoding error probability tends to zero; hence, by \Cref{def:achievability}, $R$ is achievable.

Our proof approach adapts the clustering-matching framework from the noisy drawing channel~\cite{lenz2022noisy}. Using \cref{clusteringalgo}, the decoder first groups the output sequences $\left(Z_{i,j,k} \right)_{i \in [M], j \in [Q_{1,i}], k \in [Q_{2,i,j}]}$ into $M$ primary clusters $\vhz = (\vhz_i)_{i \in [M]}$ based on Hamming distance, then refines each $\vhz_i$ into $(\vhz_{i,j})_{j \in [\hq_{1,i}]}$ sub-clusters. Excess or missing clusters are adjusted to ensure exactly $M$ primary clusters. This clustering algorithm prioritizes analytical tractability over computational efficiency and clustering accuracy, but it suffices for our information-theoretic analysis. We establish that each cluster typically contains noisy sequences originating from a single input sequence. The decoder then employs joint typicality with respect to the underlying synthesis--sequencing channel without shuffling, in the sense introduced by Shannon~\cite{Cover2006,shannon}, to pair input sequences $X_i$ with output clusters $\vhz_j$. For rates $R$ below the capacity of the synthesis--sequencing BSC channel, we show that the correct transmitted codeword produces almost $M$ matches with high probability, while incorrect codewords produce few matches with high probability. 

\begin{remark}\label{rem:subclustering}
    Informally, to guarantee high-probability correctness of the sub-clustering step in \Cref{clusteringalgo}, we require that the sub-clustering radius $\phi$ must satisfy (1) $\phi > 2 \pseq$, to ensure that sequences within the same sub-cluster are grouped together, and (2) $\phi < 2 \peff (1-\peff)$, to ensure that sequences within the same cluster but from different sub-clusters are not grouped together. These two conditions together imply $\pseq < \peff (1-\peff)$, which in turn yields $\pseq < \frac{\sqrt{\psyn(1-\psyn)}}{1+2\sqrt{\psyn(1-\psyn)}}$.
\end{remark}

To this end, we proceed with the definition of typicality.

\begin{definition}
    Consider the $\twobsc$ channel with fixed sampling vector $\vec{q} = \left(q_1,\left( q_{2,1},\dots,q_{2,q_1}\right)\right)$ and uniform input $X \in \{0,1\}$. Denoting by $n = \sum_{i=1}^{q_1}q_{2,i}$, then the corresponding output $\vec{Z} \in \{0,1\}^{n}$. We define the set of $\eps$-jointly typical sequences $x \in \{0,1\}^L$ and $\vec{z} \in \{0,1\}^{nL}$ by 
\begin{align*}
\cT_{\twobsc}^{(L,\eps)}(\psyn,\pseq,\vec{q}) \coloneqq \Big\{ (x,\vec{z}): & \Big|-\frac{\log \Pr(\vec{z})}{L} - H(\vec{Z})\Big| < \eps,   \Big|-\frac{\log \Pr(x,\vec{z})}{L} - H(X,\vec{Z})\Big| < \eps \Big\}.
\end{align*}
\end{definition}

Note that we have not considered the condition on the input $\Pr(x)$, since this is trivially fulfilled given that we consider a uniform distribution on input sequences. 
We can then define a measure of typicality over a parallel $\twobsc$ channel as follows.

\begin{definition}
    We define the largest typicality matching $\cT_{\SS}^\eps(\vec{x}, \vec{z})$ between an input $\vec{x} = (x_1,\ldots, x_M)$ and output
$\vec{z} = (\vec{z}_1,\ldots, \vec{z}_M)$ as the largest integer $T$ such that there
exist two sequences $i_1, . . . , i_T$ and $j_1, . . . , j_T$ , with
$i_t, j_t \in [M]$ for all $1 \leq t \leq T$ , each sequence composed of
distinct integers, such that $(x_{i_t} , \vec{z}_{j_t}) \in \cT_{\twobsc}(\psyn,\pseq,\vec{q_{j_t}})$ for all
$1 \leq t \leq T$, where $\vec{q_{j_t}} = \left(q_{1,j_t},(q_{2,j_t,1},\ldots,q_{2,j_t,q_{1,j_t}})\right)$. Here, $q_{1,j_t}$ is the size of the cluster $\vec{z}_{j_t}$ and for all $k \in [q_{1,j_t}]$, $q_{2,j_t,k}$ is the size of the subcluster $\vec{z}_{j_t,k}$.
\end{definition}
    
The typicality between the input sequences $\vec{X}$ and the output clusters $\vec{Z}$ is characterized by the number of distinct pairs $(X_i, \vec{Z}_j)$ that are jointly typical with respect to the $\twobsc$ channel. We let $\cC = \{\vec{X}(1),\ldots,\vec{X}(\lceil 2^{ MLR} \rceil)\} \subseteq \bits^{ML}$ be a randomly chosen codebook of code rate $R$, where each codeword $\vec{X}(w) \in \bits^{ML}$ is selected independently and uniformly over all possible words. We will write $\vec{X}(w) = (X_1(w),\ldots,X_{M}(w))$. The decoder first estimates the clusters $\vhz$ using \Cref{clusteringalgo} and then declares message $\hw$ if $\vec{X}(\hw)$ is the unique codeword satisfying $\cT_{\SS}^{\eps}(\vec{X}(\hw), \vhz) \geq M(1-\eps)$. If no codeword, or more than one codeword, satisfies this condition, the decoder declares an error.

We fix $0 < \eps < 1$ and a clustering radius $2\peff <\tau < 1/2$ and a sub-clustering radius $2\pseq <\phi < 2\peff(1-\peff)$. The proof of the achievability bound relies on intermediate claims established later. %For our analysis, we assume that our sampling distributions $\Qsyn$ and $\Qseq$ have bounded means, which is common in practice.

%\subsection{Proof of achievability bound}
\subsection{Proof of achievability bound}

For fixed $w \in [\lceil 2^{MLR} \rceil]$, we denote by $\Pr(\Err\mid W = w)$ the probability of a decoding error for the message $w$, averaged over all equiprobable codebooks. The average decoding error probability, averaged over all equiprobable codebooks, is then given by
\begin{equation*}
    \frac{1}{2^{ML}} \sum_{\cC} P_e(\cC) = \Pr(\Err\mid W = 1),
\end{equation*} %\snote{end}
by the symmetry of the random codebook construction. The two error events are: $\vec{X}(1)$ fails to be jointly typical with $\vhz$, or some other codeword is jointly typical with $\vhz$. For fixed $\eps>0$ and given that $W=1$, let $\cJ_w$ be the event that $\cT_{\SS}^\eps(\vec{X}(w),\vhz) \geq M(1-\eps/2)$ and let $\cJ_{w}^{c}$ be its complement. We use the slightly stricter threshold \(M(1-\eps/2)\) here to leave a margin of \(M\eps/2\) for the clustering errors that will be accounted for later. The union bound then gives 
\begin{align*}
    \Pr(\Err\mid W=1) & \leq \Pr\left(\cJ_1^c \cup \bigcup_{w=2}^{\lceil 2^{MLR} \rceil}\cJ_w \right) \\ & \leq \Pr(\cJ_1^c) + \Pr\left( \bigcup_{w=2}^{\lceil 2^{MLR} \rceil}\cJ_w \right).
\end{align*}
We first show that $\Pr(\cJ_1^c) \to 0$. Consider the bipartite graph $G_{\emph{cluster}}$ whose left vertices are the true clusters $\vec{Z}_i, i \in [M]$ and whose right vertices are the estimated clusters $\vhz_j, j \in [M]$ from \Cref{clusteringalgo}. We draw an edge from $\vec{Z}_i$ to $\vhz_j$ if they are equal as multisets of sub-clusters (where sub-clusters are viewed as multisets of sequences).
%their multisets of sequences coincide.
Let $G$ denote the size of the largest matching in the graph $G_{\emph{cluster}}$. Further let $\cG$ be the event that $G \geq M(1-\frac{\eps}{2})$. Then,
\begin{equation}\label{eq:first_cc_bound}
\Pr(\cJ_1^c) \leq \Pr(\cJ_1^c \mid \cG) + 1 - \Pr(\cG).    
\end{equation}

For simplicity, let $(\vec{Z}_{i}(j))_{j \in [N_i]}$ be the sequences of a cluster $\vec{Z}_i$. Further, let $(\vec{Z}_{i,j}(k))_{k \in [Q_{2,i,j}]}$ be the sequences of a subcluster $\vec{Z}_{i,j}$. By \Cref{clusteringalgo}, a received output sequence $Z_{i,j,k}, i \in [M], j \in [Q_{1,i}], k \in [Q_{2,i,j}]$ belongs to cluster $\vhz_i$ if it satisfies $d_{H}(Z_{i,j,k}, \vhz_i(1)) < \tau L$. Similarly, the output sequence $Z_{i,j,k}$ belongs to subcluster $\vhz_{i,j}$ if it satisfies $d_{H}(Z_{i,j,k}, \vhz_{i,j}(1)) < \phi L$. We assume $ 2\peff< \tau < \frac{1}{2}$, $h(\tau) < 1- \frac{1}{\beta}$, and $2\pseq < \phi < 2\peff(1-\peff)$ for the following claim regarding the accuracy of the proposed clustering algorithm~\cref{clusteringalgo}.
\begin{claim} \label{claim:clustering}
    For any fixed $\tau$ with $ 2\peff< \tau < \frac{1}{2}$ and $h(\tau) < 1- \frac{1}{\beta}$, and any fixed $\phi$ with $2\pseq < \phi < 2\peff(1-\peff)$,
    \begin{equation*}
        \Pr(G \geq M(1-\eps)) \to 1 \quad \text{as } M \to \infty,
    \end{equation*}
%$$\Pr(G \geq M(1-\eps)) \to 1 \quad \text{as } M \to \infty,$$
\end{claim}

Note that such a $\tau$ and $\phi$ is guaranteed to exist by the conditions in \Cref{thm:achievability}, namely $\peff< 1/4$, $h(2\peff) < 1 - \frac{1}{\beta}$ and $\pseq < \peff(1-\peff)$. Since \Cref{claim:clustering} holds for every $\eps>0$, it also holds with $\eps$ replaced by $\eps/2$; consequently, with high probability we have $G \ge M(1-\eps/2)$, so the clustering step fails on at most $M\eps/2$ clusters. Next, we state the following claim on the largest typicality matching between input sequences $\vec{X}$ and estimated clusters $\vhz$. We do not include its proof, as the argument is straightforward and analogous to \cite[Lemma 15]{lenz2022noisy}.

\begin{claim}\label{claim:typicality}
    The largest typicality matching between $\vec{X}$ and $\vhz$ satisfies
     %The joint typicality of $X$ and $\vhz$ satisfies 
    \begin{align*}
        |\cT_{\SS}^\eps(\vec{X},\vhz) - \cT_{\SS}^\eps(\vec{X},\vec{Z})| \leq M - G.
    \end{align*}
\end{claim}

Since \Cref{claim:clustering} shows that $G \geq M(1-\eps/2)$ with high probability, \Cref{claim:typicality} further implies that $|\cT_{\SS}(\vec{X}(w), \vhz) - \cT_{\SS}(\vec{X}(w), \vec{Z})| \leq M\eps/2$ with high probability. Plugging \Cref{claim:typicality} into \Cref{eq:first_cc_bound}, we have
\begin{align}\label{eq:second_cc_bound}
    \Pr(\cJ_1^c) & \leq \Pr\left( \cT_{\SS}^{\eps}(\vec{X}(1),\vec{Z}) < M(1-\eps/2) \mid W=1, \cG\right) + 1 - \Pr(\cG) \notag \\ & \leq \frac{\Pr\left( \cT_{\SS}^{\eps}(\vec{X}(1),\vec{Z}) < M(1-\eps/2) \mid W=1\right)}{ \Pr(\cG)} + 1 - \Pr(\cG) .
\end{align}
We then invoke standard joint typicality arguments to show that, with high probability, the decoder successfully recovers the correct input sequence.

\begin{claim} \label{claim:correct}
    For any fixed $0<\eps<1$, as $M \to \infty$, it holds that
    \begin{align*}
        \Pr\left(\cT_{\SS}^\eps(\vec{X}(1),\vec{Z}) \geq M(1-\eps)\mid W = 1 \right) \to 1.
    \end{align*}
\end{claim}

Plugging \Cref{claim:clustering,claim:correct} into \Cref{eq:second_cc_bound}, we have
\begin{equation*}
     \Pr(\cJ_1^c) = o(1).
\end{equation*}
Next, we show that $\Pr\left( \bigcup_{w=2}^{\lceil 2^{MLR} \rceil}\cJ_w \right) \to 0$ in a similar fashion. 
\begin{equation}\label{eq:first_wc_bound}
\Pr\left( \bigcup_{w=2}^{\lceil 2^{MLR} \rceil}\cJ_w \right)
 \leq \Pr\left( \bigcup_{w=2}^{\lceil 2^{MLR} \rceil}\cJ_w  \mid \cG \right) + 1 - \Pr(\cG) .   
\end{equation}
Plugging \Cref{claim:typicality} into \Cref{eq:first_wc_bound}, we have
\begin{align} \label{eq:second_wc_bound}
\Pr\left( \bigcup_{w=2}^{\lceil 2^{MLR} \rceil}\cJ_w  \right) & \leq \Pr\left( \bigcup_{w=2}^{\lceil 2^{MLR} \rceil} \cT_{\SS}^{\eps}(\vec{X}(w),\vec{Z}) \geq M(1-3\eps/2) \mid W=1, \cG\right) + 1 - \Pr(\cG) \notag \\
& \leq \frac{\Pr\left( \bigcup_{w=2}^{\lceil 2^{MLR} \rceil} \cT_{\SS}^{\eps}(\vec{X}(w),\vec{Z}) \geq M(1-3\eps/2) \mid W=1\right)}{ \Pr(\cG)} + 1 - \Pr(\cG) .
\end{align}
The next claim shows that incorrect codewords $\vec{X}(w), w \geq 2$ have few typical matches with $\vec{Z}$ with high probability.

\begin{claim}\label{claim:wrong}
    For any fixed $0<\eps<1$, and any fixed $$R <\underset{\vec{q} \sim (\Qsyn,\Qseq^{\Qsyn})}{\E}\left[ \Cbsc(\psyn,\pseq, \vec{q}) \right] - \frac{1}{\beta}(1-\qerase)  - \frac{5\eps}{2},$$ as $M \to \infty$, it holds that
    \begin{align*}
        \Pr\left( \exists w: 2\leq w \leq \lceil 2^{MLR} \rceil, \cT_{\SS}^\eps(\vec{X}(w),\vec{Z}) \geq M(1-\eps)\mid W = 1 \right) \to 0.
    \end{align*}
\end{claim}

Plugging \Cref{claim:clustering,claim:wrong} into \Cref{eq:second_wc_bound}, we have
\begin{equation*}
     \Pr\left( \bigcup_{w=2}^{\lceil 2^{MLR} \rceil}\cJ_w  \right) = o(1).
\end{equation*}
Thus, it follows that for any $\eps > 0$ and $R < \underset{\vec{q} \sim (\Qsyn,\Qseq^{\Qsyn})}{\E}\left[ \Cbsc(\psyn,\pseq, \vec{q}) \right] - \frac{1}{\beta}(1-\qerase) - \frac{5\eps}{2}$, the error probability vanishes, $\Pr(\Err\mid W=1) \to 0$ as $M \to \infty$.

Since $\eps$ can be made arbitrarily small, for any $R < \underset{\vec{q} \sim (\Qsyn,\Qseq^{\Qsyn})}{\E}\left[ \Cbsc(\psyn,\pseq, \vec{q}) \right] - \frac{1}{\beta}(1-\qerase)$ we have $\Pr(\Err\mid W=1) \to 0$ as $M \to \infty$. 
Therefore, $R$ is achievable.

%Since $\eps > 0$ can be made arbitrarily small, the result holds for all $R < \underset{\vec{q} \sim (\Qsyn,\Qseq^{\Qsyn})}{\E}\left[ \Cbsc(\psyn,\pseq, \vec{q}) \right] - \frac{1}{\beta}(1-\qerase)$. As the average error probability over random codebooks vanishes, there exists at least one codebook of rate $R$ with $\Pr(\Err \mid W=1) \to 0$ as $M \to \infty$. Thus, $R$ is an achievable rate, which concludes the proof.

\subsection{Proofs of intermediate claims}

Below, we provide the proofs of the claims used to show the achievability bound.

\begin{proof}[Proof of \Cref{claim:clustering}] 
 Let $$G_i := \mathbf{1}\left\{N_i>0,\ \exists\, j\in[M],\ \exists\, \sigma:[Q_{1,j}]\to[Q_{1,i}] \text{ s.t. } \hat{\mathbf{Z}}_{j,\ell} = \mathbf{Z}_{i,\sigma(\ell)}\ \forall \ell\in[Q_{1,j}]\right\},$$ be a binary indicator variable that is equal to $1$ when the input sequence $X_i$ has at least one output and $\vec{Z}_i$ has been clustered (and sub-clustered) correctly using \cref{clusteringalgo}. Further, let %$\hm := \sum_{i=1}^M \mathbf 1\{\hat{|\vec{Z}}_i|>0 \}$
  $\hm$ be the number of non-empty estimated clusters before removing clusters or adding empty clusters. %\snote{why do we say ``before removing or adding clusters''? the definition $\hm := \sum_{i=1}^M \mathbf 1\{\hat{|\vec{Z}}_i|>0 \}$  implies $\hm \leq M$, but it's possible that the algorithm first produces $\hm > M$ clusters and only then removes them, no? so this definition of $\hm$ seems to me like the number of nonempty clusters at the very end of the algorithm, after adding or removing clusters}\knote{because after removing or adding clusters, some estimated clusters from the algorithm might be missing and so the analysis won't be accurate. Indeed, $\hm$ can be greater than $M$. Hence, I am in favour of removing the definition $\hm := \sum_{i=1}^M \mathbf 1\{\hat{|\vec{Z}}_i|>0 \}$. Moreover, before adding or removing clusters, the total number of estimated clusters need not be M, and hence the definition of $\hat{M}_0$ is not correct either. I have replaced $\hat{M}_0$ with $M-\hm$ everywhere in the proof}%, and $\hat{M}_0 := M-\hm$ be the number of empty estimated clusters. 
 Recall that \mbox{$M_0 = |\{ i \in [M] : N_i = 0 \}|$} is the number of input sequences for which no channel output is produced, i.e., the true number of empty clusters. It holds that \mbox{$G \geq \sum_{i=1}^M G_i + \min\{M_0, M - \hm\}$}, since we can construct a matching by pairing each cluster with $G_i = 1$ to its correct counterpart in the algorithm's output, and arbitrarily matching at least $\min\{M_0, M - \hm\}$ empty clusters among the remaining unmatched ones. Using $\min\{M_0,M - \hm\} \geq M_0 - |M_0 - (M - \hm)|$, it can be verified that the event $\{G < M(1-\eps)\}$ is contained in the union of the two events $\left\{\sum_{i=1}^{M} G_i \leq M-M_0-M\eps/2\right\}$ and $\{|M_0 - (M-\hm)| \geq M\eps/2\}$, and hence by a union bound we get
\begin{align}
\Pr(G \geq M(1-\eps)) &\geq 1 - \Pr\left(\sum_{i=1}^{M} G_i \leq M-M_0-M\eps/2\right) - \Pr(|M_0 - (M-\hm)| \geq M\eps/2), \label{eq:lemmalargestmatching}\\
&= 1 - \Pr\left(M-M_0-\sum_{i=1}^{M} G_i \geq M\eps/2\right) - \Pr\left(|M_0 - (M-\hm)| \geq M\eps/{2}\right). \label{eq2:lemmalargestmatching}
\end{align}
We prove next that as $M\to\infty$,
$$\frac{1}{M}\, \E\left[ M-M_0-\sum_{i=1}^{M} G_i \right] \to 0, \quad \text{and} \quad \frac{1}{M}\, \E\left[ |M_0 - (M-\hm)| \right] \to 0,$$
and hence by applying Markov's inequality, it can be shown that both probability terms in \Cref{eq2:lemmalargestmatching} vanish asymptotically.

Let \(K_M:=\lceil \log M\rceil\), and call the \(i\)-th cluster
\emph{light} if \(N_i\le K_M\), and \emph{heavy} otherwise. Define the indices of the heavy clusters by $\mathcal H_M \coloneqq \{i\in[M]:N_i>K_M\}$. We will only attempt to certify correct clustering for light clusters. Heavy clusters will be treated pessimistically as bad clusters in the analysis, i.e., we will assume $G_i=0$ for all $i \in \mathcal H_M$. This does not mean that the clustering algorithm necessarily fails on them; rather, it means that we do not use any correctness guarantee for them in our analysis.

We have
\begin{align}
    M-M_0-\sum_{i=1}^M G_i
    &=
    \sum_{i=1}^M \mathbf 1\{N_i>0\}(1-G_i)\\
    &\le
    \sum_{i=1}^M \mathbf 1\{0<N_i\le K_M\}(1-G_i)
    +
    |\mathcal H_M|,
\end{align}
where the first term counts light non-empty clusters that are incorrectly clustered,
while the second term upper-bounds the loss due to the omission of heavy clusters. Therefore, using the fact that $N_i$'s are identically distributed, we can write 
\begin{align}
    \frac{1}{M}\, \E\left[ M-M_0-\sum_{i=1}^{M} G_i \right] &\leq \frac{1}{M}\, \sum_{i=1}^M \E\left[\mathbf 1\{0<N_i\le K_M\}(1-G_i)\right] + \frac{1}{M}\,\E\left[|\mathcal H_M| \right], \label{eqs23}\\
    &= \E\left[\mathbf 1\{0<N_1\le K_M\}(1-G_1)\right] + \frac{1}{M}\,\E\left[|\mathcal H_M| \right], \label{eqs24}\\
    &= \Pr(0<N_1\le K_M,\;G_1=0) + \Pr(N_1>K_M). \label{eqs25}
\end{align}
Since \(K_M=\lceil \log M\rceil \to\infty\) as $M\to \infty$ and $N_1$ is assumed to have finite expectation, it follows from Markov's inequality that $\Pr(N_1>K_M)\to 0$, which shows that the second term in~\Cref{eqs25} vanishes asymptotically. The first term in~\Cref{eqs25} corresponds to the probability that at least one of the sufficient conditions for correct clustering fails for a fixed light non-empty cluster. 
\begin{comment}
Next, we show that
\begin{multline} \label{eqs26}
    \Pr(0<N_1\le K_M,\;G_1=0) \leq K_M e^{-2L(\tau/2-\peff)^2} +  K_M M\E[N_1]2^{-L(1-h(\tau))} \\ + K_M e^{-2L(\phi/2-\pseq)^2} + K_M^2 e^{-2L(2\peff(1-\peff)-\phi)^2}.
\end{multline}
\end{comment}
%\begin{comment}
Next, we derive an upper bound for this term and show that it also vanishes asymptotically as $M\to \infty$. 

Given that $N_1 > 0$, cluster $\vec{Z}_1$ is guaranteed to be estimated correctly if the following sufficient conditions hold:

\begin{enumerate}
    \item[(1a)] $d_{H}(X_1, \vec{Z}_{1}(j)) < \tau L/2$ for all $j \in [N_1]$;
    
    \item[(1b)] $d_H(\vec{Z}_1(j),\vec{Z}_{i'}(j')) \geq \tau L$ for all $i' \in [M], j \in [N_i]$ and $j' \in [N_{i'}]$ such that $i' \neq 1$;
    
    \item[(2a)] $d_H(Z_{1,j},\vec{Z}_{1,j}(k)) < \phi L / 2$ for all $j \in [Q_{1,1}]$ and $k \in [Q_{2,1,j}]$, where $Z_{1,j}$ is the $j$-th synthesized sequence coming from $X_1$;

    \item[(2b)] $d_H(\vec{Z}_{1,j}(k),\vec{Z}_{1,j'}(k')) \geq \phi L$ for all $j,j' \in [Q_{1,1}], k \in [Q_{2,1,j}]$ and $k' \in [Q_{2,1,j'}]$ such that $j' \neq j$.
\end{enumerate}

Given $0 < N_1 \leq K_M$, let $\Err_{1a}$ and $\Err_{1b}$ denote the events that condition (1a) and condition (1b), respectively, fail to hold. Note that, since $d_H(X_1,\vec{Z}_{1}(j)) \sim \Bin(L,\peff)$ for any $j$, Hoeffding's inequality implies that, for $\tau > 2\peff$, the event $d_H(X_1,\vec{Z}_{1}(j)) \geq \tau L / 2$ has probability at most $e^{-2L(\tau/2 - \peff)^2}$. A union bound over the maximum $K_M$ sequences in cluster $\vec{Z}_1$ yields
\begin{equation} \label{eq:err1a-final}
    \Pr(\Err_{1a} \mid 0< N_1 \leq K_M) \leq K_M e^{-2L(\tau/2 - \peff)^2}. 
\end{equation}
We shall bound the probability of the error event $\Err_{1b}$ by conditioning on the event that $\sum_{i'=2}^{M}N_{i'} \leq MK_M$. Note that
\begin{align}\label{eq:err1b}
 \Pr (\Err_{1b} \mid 0 < N_1 \leq K_M) & \leq \Pr\left(\sum_{i'=2}^{M}N_{i'} > MK_M\right) + \Pr\left(\Err_{1b} \mid 0< N_1 \leq K_M, \sum_{i'=2}^{M}N_{i'} \leq MK_M \right). 
\end{align}
We first show that $\sum_{i'=2}^{M}N_{i'}$ is at most $MK_M$ with high probability. Using the fact that the $N_{i'}$'s are identically distributed, we get $\E{\left[\sum_{i'=2}^{M}N_{i'}\right]}  = \sum_{i'=2}^{M}\E\left[N_{i'}\right] = (M-1)\E\left[N_{1}\right]$. By Markov's inequality, we have
\begin{align}\label{eq:markov}
\Pr\left(\sum_{i'=2}^{M}N_{i'} > MK_M\right) & \leq \frac{1}{MK_M}\E{\left[\sum_{i'=2}^{M}N_{i'}\right]} < \frac{\E[N_1]}{K_M}.
\end{align}
Furthermore, for any $j$ and $i',j'$ with $i' \neq 1$, the sequences $\vec{Z}_{1}(j)$ and $\vec{Z}_{i'}(j')$ are independent and uniformly random, since $X_1$ and $X_{i'}$ are independent and uniformly random. Thus, by~\cite[Corollary 20]{lenz2022noisy}, $\Pr(d_H(\vec{Z}_{1}(j),\vec{Z}_{i'}(j')) < \tau L) \leq 2^{-L(1-h(\tau))}$ for $\tau < 1/2$. A union bound over at most $K_M$ sequences in the cluster $\vec{Z}_1$ and at most $MK_M$ other output sequences yields
\begin{equation}\label{eq:err1b-conditioned}
  \Pr\left(\Err_{1b} \mid 0< N_1 \leq K_M, \sum_{i'=2}^{M}N_{i'} \leq MK_M \right) \leq K_{M}^2 2^{-L\left(1-h(\tau)-\frac{1}{\beta}\right)},
\end{equation}
where we used that $M = 2^{L/\beta}$.
Substituting \Cref{eq:markov,eq:err1b-conditioned} into \Cref{eq:err1b} yields
\begin{equation}\label{eq:err1b-final}
   \Pr(\Err_{1b} \mid 0 < N_1 \leq K_M) \leq \frac{\E[N_{1}]}{K_M} + K_{M}^2 2^{-L\left(1-h(\tau)-\frac{1}{\beta}\right)}.  
\end{equation}
Next, let $\Err_{2a}$ and $\Err_{2b}$ denote the events that condition (2a) and condition (2b), respectively, fail to hold. Note that, since $d_H(Z_{1,j},\vec{Z}_{1,j}(k)) \sim \Bin(L,\pseq)$ for any $j,k$, Hoeffding's inequality implies that, for $\phi > 2\pseq$, the event $d_H(Z_{1,j},\vec{Z}_{1,j}(k)) > \phi L / 2$ has probability at most $e^{-2L(\phi/2 - \pseq)^2}$. A union bound over the maximum $K_M$ sequences in cluster $\vec{Z}_1$ yields
\begin{equation} \label{eq:err2a-final}
    \Pr(\Err_{2a} \mid 0< N_1 \leq K_M) \leq K_M e^{-2L(\phi/2 - \pseq)^2}. 
\end{equation}  
Moreover, for any pair of sequences in $\vec{Z}_i$ from different sub-clusters, we have $d_H(\vec{Z}_{1,j}(k),\vec{Z}_{1,j'}(k')) \sim \Bin(L,2\peff(1-\peff))$ for any $j,j',k'$. Hoeffding's inequality then implies that, for $\pseq < \peff(1-\peff)$, the event $d_H(\vec{Z}_{1,j}(k),\vec{Z}_{1,j'}(k')) < \phi L$ has probability at most $e^{-2L(2\peff(1-\peff)-\phi)^2}$. Once again, a union bound over at most $K_M$ sequences in the cluster $Z_i$ and thus over at most $K_M^2$ sequences yields
\begin{equation} \label{eq:err2b-final}
    \Pr(\Err_{2b} \mid 0< N_1 \leq K_M) \leq K_M^2e^{-2L(2\peff(1-\peff)-\phi)^2} . 
\end{equation}  
Note that, given $0 < N_1 \leq K_M$, we have $G_1 = 0$ only if one of the error events defined above occurs. Applying the union bound then gives
\begin{multline}\label{eq:errevents}
\Pr(0<N_1\le K_M,\;G_1=0) \leq \Pr(\Err_{1a} \mid 0< N_1 \leq K_M) + \Pr(\Err_{1b} \mid 0< N_1 \leq K_M) \\ + \Pr(\Err_{2a} \mid 0< N_1 \leq K_M) + \Pr(\Err_{2b} \mid 0< N_1 \leq K_M) .   
\end{multline}
Substituting \Cref{eq:err1a-final,eq:err1b-final,eq:err2a-final,eq:err2b-final} into \Cref{eq:errevents} yields 
\begin{multline} \label{eqs34}
    \Pr(0<N_1\le K_M,\;G_1=0) \leq \frac{\E[N_1]}{K_M} + K_M e^{-2L(\tau/2-\peff)^2} + K_{M}^2 2^{-L\left(1-h(\tau)-\frac{1}{\beta}\right)} \\ + K_M e^{-2L(\phi/2-\pseq)^2} + K_M^2 e^{-2L(2\peff(1-\peff)-\phi)^2}.
\end{multline}

%\end{comment}
Since \(K_M=\lceil \log M\rceil,L=\beta\log M\) and \(\E[N_1] < \infty\), all terms in~\Cref{eqs34} go to zero as $M\to \infty$ if
\begin{equation} \label{eqs:cond}
      \tau>2\peff,
    \qquad
    2\pseq<\phi<2\peff(1-\peff),\quad \text{and} \quad h(\tau)<1-\frac{1}{\beta}.  
\end{equation}
Therefore, under the conditions in~\Cref{eqs:cond}, the first term in~\Cref{eqs25} also vanishes asymptotically, which proves that $\frac{1}{M}\, \E\left[ M-M_0-\sum_{i=1}^{M} G_i \right] \to 0$.
% $\sum_{i=1}^{M} G_i$ converges in probability to $M - M_0$ as $M\to\infty$, i.e., $\sum_{i=1}^{M} G_i \xrightarrow{\Pr} M - M_0$, and that $\hat{M}_0 \xrightarrow{\Pr} M_0$ as $M\to\infty$. Hence, both probability terms on the RHS of \Cref{eq:lemmalargestmatching} vanish asymptotically.

To complete the proof, it remains to show that $
\frac{1}{M}\, \E\left[ |M_0 - (M-\hm)| \right] \to 0$. %Since $\hat{M}_0=M-\hm$, 
Rearranging, we have $|M_0-(M-\hm)|=|\hm-(M-M_0)|$.
The quantity $M-M_0$ is the true number of non-empty clusters, while $\hm$ is
the number of non-empty estimated clusters produced by the algorithm. We again
separate light and heavy clusters.

Consider first the light non-empty clusters for which $G_i=1$. By definition
of $G_i$, each such cluster is correctly recovered as exactly one estimated
non-empty cluster. Hence these clusters contribute equally to $\hm$ and to
$M-M_0$, and therefore they do not contribute to $|\hm-(M-M_0)|$.

It remains to account for clusters that are either light non-empty clusters
with $G_i=0$, or heavy clusters. Let $A_M$ denote the number of true non-empty
clusters remaining after excluding the correctly recovered light clusters.
Then, 
\begin{equation}
    A_M
    =
    \sum_{i=1}^M \mathbf 1\{0<N_i\le K_M\}(1-G_i)
    +
    |\mathcal H_M|.
\end{equation}

Now let $\hat A_M$ denote the number of estimated non-empty clusters remaining
after excluding the estimated clusters corresponding to correctly recovered
light clusters. Each such remaining estimated cluster must contain at least
one sequence originating either from a light non-empty cluster with $G_i=0$, or
from a heavy cluster. Therefore, we have
\begin{equation}
    \hat A_M
    \leq
    K_M\sum_{i=1}^M \mathbf 1\{0<N_i\le K_M\}(1-G_i)
    +
    \sum_{i=1}^M N_i\mathbf 1\{N_i>K_M\},
    \label{eq:hatAM_bound}
\end{equation}
where the first term counts the total number of sequences originating from light
non-empty clusters with $G_i=0$ and the second term counts the number of sequences originating from heavy clusters.

Since $A_M \leq  K_M\sum_{i=1}^M \mathbf 1\{0<N_i\le K_M\}(1-G_i) + \sum_{i=1}^M N_i\mathbf 1\{N_i>K_M\}$,
\begin{comment}    
$A_M \leq \hat A_M$ \snote{I don't think this is necessarily true -- e.g., suppose no cluster is recovered correctly, and all true clusters are non-empty, but there are less than $M$ estimated non-empty clusters. Then $A_M = M > \hm = \hat{A}_M$. I guess what we mean is that $A_M$ is at most the upper bound given in \Cref{eq:hatAM_bound}? Also, we could replace $2\hat{A}_M$ by just $\hat{A}_M$ in the equation below and it would still hold, right?},
\end{comment}
we have 
\[
    |\hm-(M-M_0)|
    =
    |\hat A_M-A_M| \leq K_M\sum_{i=1}^M \mathbf 1\{0<N_i\le K_M\}(1-G_i)
    +
    \sum_{i=1}^M N_i\mathbf 1\{N_i>K_M\} .
\]
Taking expectations, dividing by $M$,
and using the fact that the $N_i$'s are identically distributed, we get
\begin{align}
    \frac{1}{M}\E\left[|\hm-(M-M_0)|\right]
    &\leq
    \frac{K_M}{M}
    \sum_{i=1}^M
    \E\left[\mathbf 1\{0<N_i\le K_M\}(1-G_i)\right] 
    +
    \frac{1}{M}
    \sum_{i=1}^M
    \E\left[N_i\mathbf 1\{N_i>K_M\}\right] \\
    &=
    K_M\Pr(0<N_1\le K_M,\;G_1=0)
    +
    \E\left[N_1\mathbf 1\{N_1>K_M\}\right].
    \label{eq:M0hatM0_bound}
\end{align}

Since \(K_M=\lceil\log M\rceil\) and \(L=\beta\log M\), it follows from~\Cref{eqs34} that the first term in~\Cref{eq:M0hatM0_bound} goes to zero as $M\to \infty$. The second term is given by
\[
    \mathbb E\left[N_1\mathbf 1\{N_1>K_M\}\right]
    =
    \sum_{n=K_M+1}^{\infty} n\Pr(N_1=n),
\]
which is the tail of a convergent non-negative series when \(K_M\to\infty\).
Therefore, it follows from dominated convergence that $\mathbb E\left[N_1\mathbf 1\{N_1>K_M\}\right]\to0$, which shows that the second term in~\Cref{eq:M0hatM0_bound} also goes to zero.

Therefore, we have shown that as $M \to \infty$
\[
\frac{1}{M}\, \E\left[ M-M_0-\sum_{i=1}^{M} G_i \right] \to 0,
\qquad
\frac{1}{M}\, \E\left[ |M_0 - (M-\hm)| \right] \to 0,
\]
and therefore applying Markov's inequality to the respective terms in~\Cref{eq2:lemmalargestmatching} completes the proof.

%To conclude the proof, it remains to show that
%\[
%\frac{1}{M}\, \E\left[ |M_0 - \hat{M}_0| \right] \to 0.
%\]
\end{proof}

\begin{proof}[Proof of \Cref{claim:correct}]

 We demarginalize with respect to the sampling distribution,
\begin{equation}\label{eq:correct_decoding_conditional}
\Pr\left(\cT_{\SS}^\eps(\vec{X}(1),\vec{Z}) \geq M(1-\eps)\mid W = 1 \right) = \sum_{\vec{q}}\Pr(\vec{Q} = \vec{q}) \cdot \Pr\left(\cT_{\SS}^\eps(\vec{X}(1),\vec{Z}) \geq M(1-\eps)\mid W = 1, \vec{q} \right). 
\end{equation}
Next, we will analyze the number
\begin{align*}
    \cT_{\OSS}^\eps(\vec{X}(1),\vec{Z}) \coloneqq |\{i \in [M]: (X_i(1),\vec{Z}_i) \in \cT_{\twobsc}^{(L,\eps)}(\psyn,\pseq,\vec{q}_i)\}|
\end{align*}
of ordered jointly typical pairs over the $\twobsc$ channel. This is because 
any pair $(X_i(1),\vec{Z}_i)$ that contributes to $\cT_{\OSS}^{\eps}(\vec{X}(1),\vec{Z})$ also contributes to $\cT_{\SS}^{\eps}(\vec{X}(1),\vec{Z})$. Therefore, we have $\cT_{\SS}^{\eps}(\vec{X}(1),\vec{Z}) \geq \cT_{\OSS}^{\eps}(\vec{X}(1),\vec{Z})$. Consequently,
\begin{align}\label{eq:ordered_matching_bound}
    \Pr\left(\cT_{\SS}^\eps(\vec{X}(1),\vec{Z}) \geq M(1-\eps)\mid W = 1, \vec{q} \right) \geq \Pr\left(\cT_{\OSS}^\eps(\vec{X}(1),\vec{Z}) \geq M(1-\eps)\mid W = 1, \vec{q} \right) .
\end{align}
For a fixed sampling vector $\vec{q}$, $\cT_{\OSS}^{\eps}(\vec{X}(1),\vec{Z})$ is the sum of $M$ independent Bernoulli random variables with success probabilities
\begin{align*}
    \pi_i \coloneqq \Pr\left((X_i(1),\vec{Z}_i) \in \cT_{\twobsc}^{(L,\eps)}(\psyn,\pseq,\vec{q}_i)\mid W=1,\vec{Q}_i = \vec{q}_i \right) .
\end{align*}
From standard results on jointly typical sequences~\cite[Theorem 7.6.1]{Cover2006}, for all $\eps > 0$, $i \in [M]$ and any realization $N_i = n_i$, it holds that $\pi_i > 1 - \eps/2$ for all $L \geq L_{n_i} = O(\log n_i)$, since $\vec{Z}_i$ is obtained by transmitting $X_i(1)$ over the $\twobsc$ channel and the error probability for jointly typical decoding over the BSC decays exponentially in the blocklength~\cite{Cover2006}. 

Let \(K_M =\lceil \log M\rceil\). By Markov's inequality and the fact that $N_i$'s are identically distributed, 
\[\Pr(N_i > K_M) < \frac{\E[N_1]}{K_M}.\] 
Hence the expected number of indices $i \in [M]$ such that $N_i > K_M$ satisfies
\begin{equation*}
\E[|\{i \in [M] : N_i > K_M\}|] < \frac{M \E[N_1]}{K_M}.
\end{equation*}
Applying Markov's inequality again, we obtain
\begin{equation*}
\Pr(|\{i \in [M] : N_i > K_M\}| \geq \frac{M\eps}{2}) < \frac{2\E[N_1]}{\eps K_M}. 
\end{equation*}
Therefore, with probability at least $1 - \frac{2\E[N_1]}{\eps K_M}$, at least $M(1 - \eps/2)$ indices have $N_i \leq K_M$, and for all such indices we have $\pi_i > 1 - \eps/2$ provided that $L = L(M)$ is chosen so that $L(M) \geq \max_{0 \leq N \leq K_M} L_{{N}}$. This choice requires that $L(M) = O(\log K_M) = O(\log\log M)$. In our setting, $L = \beta \log M$ satisfies this requirement. Therefore,
\begin{align}\label{eq:ordered_matching_bound_two}
   & \Pr\left(\cT_{\OSS}^\eps(\vec{X}(1),\vec{Z}) \geq M(1-\eps)\mid W = 1, \vec{q} \right) \notag 
   \\ & \geq  \left(1-\frac{2\E[N_1]}{\eps K_M}\right)\cdot 
     \sum_{i=M-M\eps}^{M-\frac{M\eps}{2}} \binom{M - \frac{M\eps}{2}}{i} \left(1 - \frac{\eps}{2}\right)^i \left(\frac{\eps}{2}\right)^{M-\frac{M\eps}{2} - i} \notag
     \\ &
= \left(1-\frac{2\E[N_1]}{\eps K_M}\right)\cdot \sum_{i=0}^{\frac{M\eps}{2}} \binom{M - \frac{M\eps}{2}}{i} \left(1 - \frac{\eps}{2}\right)^{M-\frac{M\eps}{2} - i} \left(\frac{\eps}{2}\right)^{i} \notag
\\ & = \left(1-\frac{2\E[N_1]}{\eps K_M}\right)\cdot\left(1 - \sum_{i= \frac{M\eps}{2} +1}^{M-\frac{M\eps}{2}} \binom{M - \frac{M\eps}{2}}{i} \left(1 - \frac{\eps}{2}\right)^{M-\frac{M\eps}{2} - i} \left(\frac{\eps}{2}\right)^i\right) \notag
\\ &
\overset{(a)}{\geq} \left(1 - \frac{2\E[N_1]}{\eps K_M}\right)\cdot \left(1- e^{-2(M-\frac{M\eps}{2}) (\frac{\eps^2}{ 4 - 2\eps})^2}\right),
\end{align}
where we used Hoeffding's inequality for the binomial tail with deviation 
\(\frac{\eps}{2(1-\eps/2)} - \frac{\eps}{2} = \frac{\eps^2}{4-2\eps}\) in inequality~(a). Substituting \Cref{eq:ordered_matching_bound,eq:ordered_matching_bound_two} into \Cref{eq:correct_decoding_conditional} gives
\begin{align*}
    \Pr\left(\cT_{\SS}^\eps(\vec{X}(1),\vec{Z}) \geq M(1-\eps)\mid W = 1 \right) & \geq \sum_{\vec{q}}\Pr(\vec{Q} = \vec{q}) \cdot \left(1 - \frac{2\E[N_1]}{\eps K_M}\right)\cdot \left(1- e^{-2(M-\frac{M\eps}{2}) (\frac{\eps^2}{ 4 - 2\eps})^2}\right) \\ & = \left(1 - \frac{2\E[N_1]}{\eps K_M}\right)\cdot \left(1- e^{-2(M-\frac{M\eps}{2}) (\frac{\eps^2}{ 4 - 2\eps})^2}\right),
\end{align*}
which approaches $1$ as $M \to \infty$ for any $0 < \eps < 1$, since $K_M \to \infty$ and the exponential term tends to $0$. Hence, the probability of correct decoding converges to $1$, proving the claim.
\end{proof}

\begin{proof}[Proof of \Cref{claim:wrong}]
    Given that $W=1$, denote by $\cJ_w'$ the event that $\cT_{\SS}^{\eps}(\vec{X}(w),\vec{Z}) \geq M(1-\eps)$. We demarginalize with respect to the sampling distribution, 
\begin{align}\label{eq:wrongerrorprobability}
   \Pr\left(\cup_{w=2}^{\lceil 2^{MLR} \rceil}\cJ_w'\right) & \leq \sum_{\vec{q}} \Pr(\vec{Q} = \vec{q})\cdot \Pr\left(\cup_{w=2}^{\lceil 2^{MLR} \rceil}\cJ_w'\mid \vec{q} \right) \notag \\ & \overset{(a)}{\leq} \lceil 2^{MLR} \rceil \sum_{\vec{q}}\Pr(\vec{Q} = \vec{q}) \cdot \Pr\left(\cJ_2'\mid  \vec{q}\right),  
\end{align}
In inequality $(a)$, we applied the union bound, leveraging the fact that $\Pr(\cJ'_w \mid \vec{q})$ remains identical for all $2 \leq w \leq \lceil 2^{MLR} \rceil$ due to the i.i.d. random codebook construction. To this end, let $\cP(M,h)$ denote the set of all length-$h$ \emph{partial permutations} of $[M]$. Further, let $T_{i,j}$ be a binary indicator variable that is equal to $1$, if $(X_{i}(2),\vec{Z}_j) \in \cT_{\twobsc}^{(L,\eps)}(\psyn,\pseq,\vec{q}_i)$. Since $W=1$, the codeword $\vec{X}(2)$ is independent of $\vec{Z}$, and thus the indicators $T_{i,j}$ are for independently chosen sequences. This enables us to rewrite the above probability as
\begin{align*}
   \Pr\left(\cJ_2'\mid \vec{q} \right) = \Pr\left(\exists (j_1,\ldots,j_{M}) \in \cP(M,M): \sum_{i=1}^{M}T_{i,j_i} \geq M(1-\eps) \mid W = 1, \vec{q} \right) .
\end{align*}
 Note that for any empty clusters, i.e., $i \in [M]$ with $q_i = 0$, we have $T_{i,j} = 1$ with probability $1$. Let $i_1, \dots, i_{M-m_0(\vec{q})}$ denote the indices where $q_{i_t} > 0$ for all $1 \leq t \leq M - m_0(\vec{q})$. A union bound over all non-empty clusters yields
 \begin{align}\label{eq:J2_bound}
   \Pr\left(\cJ_2'\mid \vec{q} \right) \leq \sum_{(i_1,\ldots,i_{M-m_0(\vec{q})}) \in \cP(M,M-m_0(\vec{q}))}\Pr\left( \sum_{t=1}^{M-m_0(\vec{q})}T_{i_{t},j_{i_t}}\geq M(1-\eps) - m_0(\vec{q})\mid W = 1, \vec{q} \right) .
\end{align}
Again, using the results about jointly typical sequences~\cite{Cover2006}, we know that 
\begin{align*}
    \pi_i \coloneqq \Pr(T_{i,j}=1 \mid W = 1, \vec{Q}_i = \vec{q}_i) < 2^{-L\left(\Cbsc(\psyn,\pseq,\vec{q}_i)-\eps\right)} 
\end{align*}
for $L \geq L_{\vec{q}_i}$ since $X(2)$ is chosen independently from $Z$.

With at least $M(1-\eps)-m_0(\vec{q})$ of the $T_{i_{t},j_{i_t}}$ Bernoulli variables being 1, we can bound the probability as
\begin{align}\label{eq:T_bound}
  \Pr\left( \sum_{t=1}^{M-m_0(\vec{q})}T_{i_{t},j_{i_t}}\geq M(1-\eps) - m_0(\vec{q})\mid W = 1, \vec{q} \right)  & \leq \sum_{\cS \subseteq [M-m_0(\vec{q})]:|S| = M(1-\eps) - m_0(\vec{q})} \prod_{t \in \cS}\pi_{i_t} \notag \\ & \overset{(b)}{\leq} \sum_{\cS \subseteq [M]:|S| \leq M(1-\eps)} \prod_{i \in \cS}\pi_{i} \notag \\ & \leq \binom{M}{M(1-\eps)}\max_{\cS \subseteq [M]:|\cS| = M(1-\eps)} \prod_{i \in \cS} \pi_i .
\end{align}
In inequality $(b)$, we included the $i$ with $Q_i = 0$ in the product to streamline subsequent notation and analysis. This only increases the set of subsets we consider, so the bound remains valid. To this end, following the proof of \Cref{claim:correct}, we let $K_M = \lceil \log M \rceil$ and abbreviate $\cK(M) = \{i \in [M]: N_i \leq K_M \}$. From \Cref{claim:correct}, with probability at least $1 - \frac{2\E[N_1]}{\eps K_M}$, we have $|\cK(M)| \geq M(1 - \frac{\eps}{2})$. Conditioning on this event, we bound the product over $\pi_i$ as
\begin{align} \label{eq:pi_bound}
    \prod_{i \in \cS} \pi_i & \leq  \prod_{i \in \cS \cap \cK(M)} \pi_i < \prod_{i \in \cS \cap \cK(M)} 2^{-L\left(\Cbsc(\psyn,\pseq,\vec{q}_i)-\eps\right)} \notag \\ & = 2^{-L\sum_{j \in \cS \cap \cK(M)}\left(\Cbsc(\psyn,\pseq,\vec{q}_i)-\eps\right)}  
\end{align}
for all $L(M) \geq \max_{0 \leq N \leq K_M}L_{N}$. Analyzing the exponent term in the expression above, we have
\begin{align}\label{eq:exponent_bound}
&\sum_{i \in \cS \cap \cK(M)}\left(\Cbsc(\psyn,\pseq,\vec{q}_i)-\eps\right) \notag \\
& = \sum_{i = 1}^{M}\left(\Cbsc(\psyn,\pseq,\vec{q}_i)-\eps\right) - \sum_{i \notin \cS \cap \cK(M)}\left(\Cbsc(\psyn,\pseq,\vec{q}_i)-\eps\right) \notag \\ & \overset{(c)}{\geq} \sum_{i = 1}^{M}\left(\Cbsc(\psyn,\pseq,\vec{q}_i)-\eps\right) - \frac{3M\eps}{2} \notag \\ & =  \sum_{i = 1}^{M}\Cbsc(\psyn,\pseq,\vec{q}_i) - \frac{5M\eps}{2}.
\end{align}
In inequality $(c)$, we bounded the second sum using $\Cbsc(\psyn,\pseq,\vec{q}_i) \leq 1$ together with $|\cK(M)| \geq M(1 - \frac{\eps}{2})$ and thus
\begin{align*}
    \Big|\{i \in [M]: i \notin \cS\cap\cK(M)\}\Big| \leq M - |\cS| + M - |\cK(M)|  \leq \frac{3M\eps}{2}.
\end{align*}
Substituting \Cref{eq:T_bound,eq:pi_bound,eq:exponent_bound} into \Cref{eq:J2_bound}, the resulting upper bound on $\Pr(\cJ_{2}'\mid\vec{q})$ is then
\begin{align} \label{eq:wrongerrorprobabilityfixedsampling}
    \Pr(\cJ_2'\mid \vec{q}) & \leq \hspace{-2mm} \sum_{(i_1,\ldots,i_{M-m_0(\vec{q})}) \in \cP(M,M-m_0(\vec{q}))}\left(1-\frac{2\E[N_1]}{\eps K_M}\right)\cdot \binom{M}{M(1-\eps)}2^{-L\left(\sum_{i = 1}^{M}\Cbsc(\psyn,\pseq,\vec{q}_i) - \frac{5M\eps}{2} \right)} \notag 
    \\ &  = \left(1-\frac{2\E[N_1]}{\eps K_M}\right)\cdot \big| \cP(M,M-m_0(\vec{q}))\big| \binom{M}{M(1-\eps)}2^{-L\left(\sum_{i = 1}^{M}\Cbsc(\psyn,\pseq,\vec{q}_i) - \frac{5M\eps}{2} \right)} \notag 
    \\ & \overset{(d)}{\leq} \left(1-\frac{2\E[N_1]}{\eps K_M}\right)\cdot 2^{M-L\left(\sum_{i = 1}^{M}\Cbsc(\psyn,\pseq,\vec{q}_i) - \frac{1}{\beta} (M-m_0(\vec{q})) -\frac{5M\eps}{2} \right)} .
\end{align}
In inequality $(d)$, we have used $\binom{M}{M(1-\eps)} \leq 2^M$ and \[|\cP(M,M-m_0(\vec{q}))| = \frac{M!}{m_0(\vec{q})!} \leq M^{M-m_0(\vec{q})} \leq 2^{ \frac{L}{\beta}(M-m_0(\vec{q}))}.\]
Substituting \Cref{eq:wrongerrorprobabilityfixedsampling} into  \Cref{eq:wrongerrorprobability}, we obtain
\begin{align*}
   & \Pr\left(\cup_{w=2}^{\lceil 2^{MLR} \rceil}\cJ_w'\right) \\ & \leq \lceil 2^{MLR} \rceil \sum_{\vec{q}} \Pr(\vec{Q} = \vec{q})\cdot \left(1-\frac{2\E[N_1]}{\eps K_M}\right) \cdot 2^{M-L\left(\sum_{i = 1}^{M}\Cbsc(\psyn,\pseq,\vec{q}_i) - \frac{1}{\beta} (M-m_0(\vec{q})) -\frac{5M\eps}{2} \right)} \\ & \overset{(e)}{=} \left(1-\frac{2\E[N_1]}{\eps K_M}\right)\cdot 2^{-ML\left(-R + \underset{ \vec{q} \sim (\Qsyn,\Qseq^{\Qsyn})}{\E}\left[ \Cbsc(\psyn,\pseq, \vec{q}) \right] - \frac{1}{\beta}(1-\qerase) - \frac{5\eps}{2} - \frac{1}{L} \right)}.
\end{align*}
In equality $(e)$, we used $\lceil 2^{MLR} \rceil \leq 2^{MLR+1} = 2^{ML(R + 1/(ML))}$, so the ceiling contributes at most $O(1/(ML))$ to the exponent, together with $\E[M_0] = M\qerase$ and the expected value of the capacity of the $\twobsc$ channel over the sampling distribution. Since $K_M \to \infty$ and $L \to \infty$ as $M \to \infty$, we have for any 
\[R < \underset{\vec{q} \sim (\Qsyn,\Qseq^{\Qsyn})}{\E}\left[ \Cbsc(\psyn,\pseq, \vec{q}) \right] - \frac{1}{\beta}(1-\qerase) - \frac{5\eps}{2},\] 
the error probability converges to $0$ as $M \to \infty$, proving the claim.
\end{proof}

\section{Capacity for the $\twobsc$ channel}\label{sec:cap_twobsc}

Let $B_{n,p}(x) \coloneqq \binom{n}{x}p^x(1-p)^{n-x}$ denote the binomial probability mass function. The capacity for the $\twobsc$ channel is given in the next theorem.

\begin{theorem}
The capacity $\Cbsc(\psyn,\pseq,\vec{q})$ for the $\twobsc$ channel with fixed sampling vector $\vec{q} = (q_1,(q_{2,1},\ldots,q_{2,q_1}))$ is given by
\small{
\begin{equation*}
    \Cbsc(\psyn,\pseq,\vec{q}) = 1 + \sum_{k=0}^{q_1} \alpha_k\sum_{s \in \cS_{q_1,k}}\sum_{(k_1,\ldots,k_{q_1})= (0,\ldots,0)}^{(q_{2,1},\ldots,q_{2,q_1})} B_{0}(q_{2,i} - k_i;s) B_{1}(k_i;s) \log f(k_1,\ldots,k_{q_1}),
\end{equation*}
}

where $\alpha_k \coloneqq (1-\psyn)^k\psyn^{q_1-k}$, $\cS_{q_1,k} \coloneqq \{(s_1,\ldots,s_{q_1}) \in \{0,1\}^{q_1}: |\{i: s_i = 0\}| = k \}$, $ B_{j}(x;s) \coloneqq \prod_{i : s_i = j}B_{q_{2,i},\pseq}(x)$ for $j \in \{0,1\}$, and $f(k_1,\ldots,k_{q_1})$ is given by
\begin{align*}
    f(k_1,\ldots,k_{q_1}) \coloneqq \frac{\sum_{k=0}^{q_1}\alpha_k\sum_{s \in \cS_{q_1,k}} B_0(q_{2,i} -k_i;s)B_1(k_i;s)}{\sum_{k=0}^{q_1}(\alpha_k + \alpha_{q_1-k})\sum_{s \in \cS_{q_1,k}}B_0(q_{2,i} -k_i;s)B_1(k_i;s)} .
\end{align*}
\end{theorem}

\begin{proof}
    We observe that the $\twobsc$ channel is memoryless across input bits and is symmetric. Given the sampling vector $\vec{q} = (q_1,(q_{2,1},\ldots,q_{2,q_1}))$, the channel outputs can be equivalently represented as follows: for an input $X = x \in\{0,1\}$, the first stage produces $y^* = \sigma(0^k1^{q_1-k})$ and the second stage produces $z^* = \sigma_1(0^{k_1}1^{q_{2,1}-k_1})\ldots\sigma_{q_1}(0^{k_{q_1}}1^{q_{2,q_1}-k_{q_1}})$, where $\sigma$ is a permutation on $[q_1]$ and $\sigma_i$ is a permutation on $[q_{2,i}]$ for each $i \in [q_1]$.
    
    To compute $I(X;Z) = H(Z) - H(Z\mid X)$, we analyze each term separately. Note that since $X$ is uniform, we have $H(Z \mid X) = H(Z \mid X = 0)$. From the channel output realization described above, we have
    \begin{align} \label{equation: condent}
        H(Z \mid X = 0) = - \sum_{(k_1,\ldots,k_{q_1}) = (0,\ldots,0)}^{(q_{2,1},\ldots,q_{2,q_1})} \prod_{i=1}^{q_1}\binom{q_{2,i}}{k_i} \Pr(Z = z^* \mid X = 0) \log \Pr(Z = z^* \mid X = 0) ,
    \end{align}
    and
    \begin{align} \label{equation: outputent}
        H(Z ) = - \sum_{(k_1,\ldots,k_{q_1}) = (0,\ldots,0)}^{(q_{2,1},\ldots,q_{2,q_1})} \prod_{i=1}^{q_1}\binom{q_{2,i}}{k_i} \Pr(Z = z^*) \log \Pr(Z = z^*) .
    \end{align}
Next, for a fixed $k \in \{0,1,\ldots,q_1\}$, consider $\alpha_k = (1-\psyn)^k\psyn^{q_1-k}$ and define 
\[\cS_{q_1,k} = \{(s_1,\ldots,s_{q_1}) \in \{0,1\}^{q_1}: |\{i: s_i = 0\}| = k\},\] 
the set of binary strings of length $q_1$ with exactly $k$ zeros. Conditioning $\Pr(Z = z^* \mid X = 0)$ on the first-stage output $Y$, we have
\begin{align} \label{equation: condprob}
    \Pr(Z = z^* \mid X = 0) &= \sum_{k=0}^{q_1}\alpha_k\sum_{s \in \cS_{q_1,k}} \Pr(Z = z^* \mid X = 0, Y = s) \notag \\ &= \sum_{k=0}^{q_1}\alpha_k\sum_{s \in \cS_{q_1,k}} \prod_{i:s_i=0}(1-\pseq)^{k_i}\pseq^{q_{2,i} - k_i}\prod_{i:s_i=1}(1-\pseq)^{q_{2,i} - k_i}\pseq^{k_i}.
\end{align}
Here, $\Pr(Z = z^* \mid X = 0)$ depends on the summation variables $(k_1,\ldots,k_{q_1})$ through the output realization $z^*$. Similarly, we have
\begin{align*}
    \Pr(Z = z^* \mid X = 1) &= \sum_{k=0}^{q_1}\alpha_{q_1-k}\sum_{s \in \cS_{q_1,k}} \prod_{i:s_i=0}(1-\pseq)^{k_i}\pseq^{q_{2,i} - k_i}\prod_{i:s_i=1}(1-\pseq)^{q_{2,i} - k_i}\pseq^{k_i}.
\end{align*}
In this expression, we still assume that the first-stage output lies in $\cS_{q_1,k}$ (i.e., it has exactly $k$ zeros), so the only difference from the $X=0$ case is the outer weight $\alpha_{q_1-k}$. Substituting \Cref{equation: condprob} into \Cref{equation: condent} and changing the order of summation, we get
\begin{equation}\label{eq:conditional_output_entropy}
H(Z \mid X = 0) = -\sum_{k=0}^{q_1}\alpha_k\sum_{s \in \cS_{q_1,k}}\sum_{(k_1,\ldots,k_{q_1}) = (0,\ldots,0)}^{(q_{2,1},\ldots,q_{2,q_1})} B_0(q_{2,i}-k_i;s)B_1(k_i;s)\log \Pr(Z = z^* \mid X = 0),
\end{equation}
where $B_0$ and $B_1$ are defined in the theorem statement. Next, we consider the output entropy $H(Z)$. Note that
\begin{align}\label{equation:outputprob}
\Pr(Z = z^*) &= \frac{1}{2}\Pr(Z = z^*\mid X = 0) + \frac{1}{2}\Pr(Z = z^* \mid X = 1)
\notag \\ & = \frac{1}{2}\sum_{k=0}^{q_1}(\alpha_k + \alpha_{q_1-k})\sum_{s \in \cS_{q_1,k}}\prod_{i:s_i=0}(1-\pseq)^{k_i}\pseq^{q_{2,i} - k_i}\prod_{i:s_i=1}(1-\pseq)^{q_{2,i} - k_i}\pseq^{k_i}.
\end{align}
Substituting \Cref{equation:outputprob} into \Cref{equation: outputent} and changing the order of summation, we get
\begin{align}\label{eq:output_entropy}
H(Z) &= \frac{1}{2}\sum_{k=0}^{q_1}(\alpha_k + \alpha_{q_1-k})\sum_{s \in \cS_{q_1,k}}\sum_{(k_1,\ldots,k_{q_1}) = (0,\ldots,0)}^{(q_{2,1},\ldots,q_{2,q_1})}B_0(q_{2,i}-k_i;s)B_1(k_i;s) \log \Pr(Z = z^*) \notag \\ &
\overset{(a)}{=} \sum_{k=0}^{q_1}\alpha_k\sum_{s \in \cS_{q_1,k}}\sum_{(k_1,\ldots,k_{q_1}) = (0,\ldots,0)}^{(q_{2,1},\ldots,q_{2,q_1})} B_0(q_{2,i}-k_i;s)B_1(k_i;s)\log\Pr(Z = z^*) \notag \\ &
\overset{(b)}{=} 1 + \sum_{k=0}^{q_1} \alpha_k\sum_{s \in \cS_{q_1,k}}\sum_{(k_1,\ldots,k_{q_1}) = (0,\ldots,0)}^{(q_{2,1},\ldots,q_{2,q_1})} B_0(q_{2,i}-k_i;s)B_1(k_i;s)\log\left(2\Pr(Z = z^*)\right).
\end{align}
In equality $(a)$, we divided the expression into two different summations
\begin{align*}
\sum_{k=0}^{q_1}\alpha_k\sum_{s \in \cS_{q_1,k}}\sum_{(k_1,\ldots,k_{q_1}) = (0,\ldots,0)}^{(q_{2,1},\ldots,q_{2,q_1})} B_0(q_{2,i}-k_i;s)B_1(k_i;s)\log\Pr(Z = z^*) ,
\end{align*}
and
\begin{align*}
\sum_{k=0}^{q_1}\alpha_{q_1-k}\sum_{s \in \cS_{q_1,k}}\sum_{(k_1,\ldots,k_{q_1}) = (0,\ldots,0)}^{(q_{2,1},\ldots,q_{2,q_1})}B_0(q_{2,i}-k_i;s)B_1(k_i;s)\log\Pr(Z = z^*).
\end{align*}
In the second sum, we change variables by mapping $k \mapsto q_1 - k$ and, for each $i$, $k_i \mapsto q_{2,i} - k_i$. 
Under this transformation, the inner product $B_0(q_{2,i}-k_i;s)B_1(k_i;s)$ and $\log\Pr(Z = z^*)$ are invariant, 
so the two sums are equal. Hence, the total equals twice the first sum, which cancels the factor $1/2$. In equality $(b)$, we have used the fact that $\sum_{k=0}^{q_1}\alpha_k\sum_{s \in \cS_{q_1,k}}\sum_{(k_1,\ldots,k_{q_1}) = (0,\ldots,0)}^{(q_{2,1},\ldots,q_{2,q_1})} B_0(q_{2,i}-k_i;s)B_1(k_i;s) = 1$. 
Therefore from \Cref{eq:conditional_output_entropy} and \Cref{eq:output_entropy}, we have
\begin{align*}
I(X;Z) &= H(Z) - H(Z\mid X) \\ &= 1 + \sum_{k=0}^{q_1} \alpha_k\sum_{s \in \cS_{q_1,k}}\sum_{(k_1,\ldots,k_{q_1})= (0,\ldots,0)}^{(q_{2,1},\ldots,q_{2,q_1})} B_0(q_{2,i}-k_i;s)B_1(k_i;s) \log f(k_1,\ldots,k_{q_1}),
\end{align*}
where $f(k_1,\ldots,k_{q_1})$ is given by
\begin{align*}
    f(k_1,\ldots,k_{q_1}) &= \frac{\Pr(Z = z^*\mid X = 0)}{2\Pr(Z = z^*)} \\ &= \frac{\sum_{k=0}^{q_1}\alpha_k\sum_{s \in \cS_{q_1,k}} B_0(q_{2,i}-k_i;s)B_1(k_i;s)}{\sum_{k=0}^{q_1}(\alpha_k + \alpha_{q_1-k})\sum_{s \in \cS_{q_1,k}}B_0(q_{2,i}-k_i;s)B_1(k_i;s)}.
\end{align*}

\end{proof}

\section*{Acknowledgment}
The work of K. Goyal and S. Kas Hanna was supported by the French government through the France 2030 investment plan managed by the National Research Agency (ANR), as part of the Initiative of Excellence Université Côte d’Azur under reference number ANR-15-IDEX-01, and by the ANR grant ANR-22-CPJ2-0054-01.

The work of S.\ Pearson and J.\ Ribeiro was funded by the European Union (LESYNCH, 101218842) and by national funds through FCT – Fundação para a Ciência e a Tecnologia, I.P., and, when eligible, co-funded by EU funds under project/support UID/50008/2025 – Instituto de Telecomunicações, with DOI \href{https://doi.org/10.54499/UID/50008/2025}{10.54499/UID/50008/2025}. Views and opinions expressed are however those of the authors only and do not necessarily reflect those of the European Union or the European Research Council Executive Agency. Neither the European Union nor the granting authority can be held responsible for them.

\bibliographystyle{IEEEtran}
\bibliography{refs}

\end{document}